\documentclass[11pt,envcountsect]{llncs}
\usepackage[margin=1in]{geometry}
\usepackage{amsmath,amssymb,times}
\usepackage[final]{graphicx}

\def\ceil#1{{\left\lceil#1\right\rceil}}

\def\tceil#1{{\lceil#1\rceil}}

\def\union{\cup}
\def\cut{\cap}
\def\eps{\varepsilon}

\newcommand{\reals}{\mathrm{I\!R}}

\def\cost{\mathrm{cost}}
\def\l{\chi}
\def\r{\chi'}

\def\med{\mathop{\mathrm{med}}}
\def\I{\mathcal{I}}
\def\Ithr{\mathcal{I}_3}
\def\IthrS{\mathcal{I}_3^-}
\def\IthrL{\mathcal{I}_3^+}
\def\sep{\!\!:\!\!}

\newtheorem{myclaim}[example]{Claim}

\spnewtheorem*{proofsketch}{Proof sketch}{\itshape}{\rmfamily}

\def\insertfig#1#2#3{%
\begin{figure}[t]
\centerline{\includegraphics[width=#1]{#2}}
\caption{#3}
\end{figure}}

\title{On the Power of Deterministic Mechanisms for Facility Location Games%
\thanks{Research partially supported by an NTUA Basic Research Grant (PEBE 2009).}}

\author{Dimitris Fotakis\inst{1} \and Christos Tzamos\inst{2}}

\institute{%
School of Electrical and Computer Engineering,\\
National Technical University of Athens, 157 80 Athens, Greece.\\
Email: {\tt fotakis@cs.ntua.gr}
\and
Computer Science and Artificial Intelligence Laboratory,\\
Massachusetts Institute of Technology, Cambridge, MA 02139.\\
Email: {\tt tzamos@mit.edu}}

\begin{document}

\maketitle

\begin{abstract}
We consider $K$-Facility Location games, where $n$ strategic agents
report their locations in a metric space, and a mechanism maps them to $K$
facilities. The agents seek to minimize their connection cost, namely the
distance of their true location to the nearest facility, and may misreport
their location. We are interested in deterministic mechanisms that are
strategyproof, i.e., ensure that no agent can benefit from misreporting her
location, do not resort to monetary transfers, and achieve a bounded
approximation ratio to the total connection cost of the agents (or to the $L_p$
norm of the connection costs, for some $p \in [1, \infty)$ or for $p = \infty$).

Our main result is an elegant characterization of deterministic strategyproof
mechanisms with a bounded approximation ratio for 2-Facility Location on the
line. In particular, we show that for instances with $n \geq 5$ agents, any such
mechanism either admits a unique dictator, or always places the facilities at
the leftmost and the rightmost location of the instance. As a corollary, we obtain that the best approximation ratio achievable by deterministic strategyproof mechanisms for the problem of locating 2 facilities on the line to minimize the total connection cost is precisely $n-2$. Another rather surprising consequence is that the {\sc Two-Extremes} mechanism of (Procaccia and Tennenholtz, EC~2009) is the only deterministic anonymous strategyproof mechanism with a bounded approximation ratio for 2-Facility Location on the line.

The proof of the characterization employs several new ideas and technical tools, which provide new insights into the behavior of deterministic strategyproof mechanisms for $K$-Facility Location games, and may be of independent interest. Employing one of these tools, we show that for every $K \geq 3$, there do not exist any deterministic anonymous strategyproof mechanisms with a bounded approximation ratio for $K$-Facility Location on the line, even for simple instances with $K+1$ agents. Moreover, building on the characterization for the line, we show that there do not exist any deterministic strategyproof mechanisms with a bounded approximation ratio for 2-Facility Location on more general metric spaces, which is true even for simple instances with 3 agents located in a star.
\end{abstract}

\thispagestyle{empty}%
\setcounter{page}{0}%
\newpage
\pagestyle{plain}
\pagenumbering{arabic}

\section{Introduction}
\label{s:intro}

We consider \emph{$K$-Facility Location games}, where $K$ facilities are placed
in a continuous metric space based on the preferences of $n$ strategic agents. Such problems are motivated by natural scenarios in Social Choice, where the government plans to build a fixed number of public facilities in an area (see e.g. \cite{Miy01}). The choice of the locations is based on the preferences of local people, or \emph{agents}. So each agent reports her ideal location, and the government applies a \emph{mechanism} mapping the agents' preferences to $K$ facility locations.
The government's objective is to minimize the \emph{social cost}, namely the total distance of the agents' locations to the nearest facility. On the other hand, the agents seek to minimize their \emph{individual cost}, namely the distance of their location to the nearest facility. In fact, an agent may even misreport her ideal location in an attempt of manipulating the mechanism. Therefore, the mechanism should be \emph{strategyproof}, i.e., should ensure that no agent can benefit from misreporting her location. 
%
%or even \emph{group strategyproof}, i.e., should ensure that for any group of agents misreporting their locations, at least one of them does not benefit.
%
Moreover, to compute a socially desirable outcome, the mechanism should achieve a reasonable approximation to the optimal social cost.

\smallskip\noindent{\bf Previous Work.}
In addition to strategyproofness, which is an essential property of any mechanism, Social Choice suggests a few additional efficiency-related properties, e.g., onto, non-dictatorship, and Pareto-efficiency, that usually accompany strategyproofness, and ensure that the mechanism's outcome is socially desirable (or at least tolerable). There are several examples of beautiful characterization theorems which state that for a particular domain, the class of strategyproof mechanisms with some efficiency-related properties coincides with a rather restricted class of mechanisms (see e.g., \cite{Bar01}). 
%
%Prominent among them is the impossibility result of %Gibbard-Satterthwaite
%\cite{Gibb73,Satter75}, which states that for general preference domains, if there are at least three possible outcomes, any onto strategyproof mechanism is a dictatorship.
%
A notable example of a problem admitting a rich class of strategyproof mechanisms is that of locating a single facility on the real line, where the agents' preferences are single-peaked. The classical characterization of Moulin \cite{Moul80} shows that the class of strategyproof mechanisms for 1-Facility Location on the line coincides with the class of generalized median voter schemes. Schummer and Vohra \cite{SV02} proved that this characterization extends to tree metrics, while for non-tree metrics, any onto strategyproof mechanism is a dictatorship. Recently, Dokow et al. \cite{DFMN12} obtained characterizations, similar in spirit to those in \cite{Moul80,DFMN12}, for the class of onto strategyproof mechanisms for 1-Facility Location on the discrete line and on the discrete circle.

Adopting an algorithmic viewpoint, Procaccia and Tennenholtz \cite{PT09} introduced the framework of \emph{approximate mechanism design without money}. The idea is to consider game-theoretic versions of optimization problems, where a social objective function summarizes (or even strengthens) the efficiency-related properties. Any reasonable approximation to the optimal solution can be regarded as a socially desirable outcome, and we seek to determine the best approximation ratio achievable by strategyproof mechanisms.
For example, the results of \cite{Moul80,SV02} imply that 1-Facility Location in tree metrics can be solved optimally by a strategyproof mechanism.
%
%that extended median voter schemes include the {\sc Median} mechanism, that places a single facility so that the total distance to the agents is minimized, and thus, 1-Facility Location in tree metrics can be solved optimally by a strategyproof mechanism.
%
On the other hand, the negative result of \cite{SV02} implies that the best approximation ratio achievable by deterministic mechanisms for 1-Facility Location in general metrics is $n-1$.

Procaccia and Tennenholtz \cite{PT09} considered several location problems on the real line, and obtained upper and lower bounds on the approximation ratio achievable by strategyproof mechanisms. For 2-Facility Location, they suggested the {\sc Two-Extremes} mechanism, that places the facilities at the leftmost and the rightmost location, and achieves an approximation ratio of $n-2$. On the negative side, they proved a lower bound of $3/2$ on the approximation ratio of any deterministic mechanism, and conjectured that the lower bound for deterministic mechanisms is $\Omega(n)$.
Subsequently, Lu et al. \cite{LWZ09} strengthened the lower bound for deterministic mechanisms to 2, established a lower bound of 1.045 for randomized mechanisms, and presented a simple randomized $n/2$-approximation mechanism.
Shortly afterwards, Lu et al. \cite{LSWZ10} significantly improved the lower bound for deterministic mechanisms to $(n-1)/2$. On the positive side, they presented a deterministic $(n-1)$-approximation mechanism for 2-Facility Location on the circle, and proved that a natural randomized mechanism, the so-called {\sc Proportional} mechanism, is strategyproof and achieves an approximation ratio of 4 for 2-Facility Location in any metric space. Lu et al. \cite{LSWZ10} observed that although {\sc Proportional} is not strategyproof for more than two facilities, its combination with {\sc Two-Extremes} results in a randomized $(n-1)$-approximation mechanism for 3-Facility Location on the line.

\smallskip\noindent{\bf Motivation and Contribution.}
Facility Location games are among the central problems in the research agenda of mechanism design without money, and have received considerable attention. Our work is motivated by the apparent difficulty of obtaining any strong(er) positive results on the approximability of $K$-Facility Location by deterministic mechanisms.
In fact, among the main open problems of \cite{LSWZ10} were (i) to determine the best approximation ratio achievable by deterministic mechanisms for 2-Facility Location on the line, (ii) to investigate the existence of deterministic mechanisms with a bounded approximation ratio for $K$-Facility Location with $K \geq 3$, and (iii) to investigate the existence of deterministic mechanisms with a bounded approximation ratio for $2$-Facility Location in metric spaces other than the line and the circle. In this work, we resolve the first question, and obtain strong negative results for the second and the third.

Attacking these questions requires a complete understanding of the behavior of deterministic strategyproof mechanisms for $K$-Facility Location, similar to that offered by characterizations in Social Choice. Hence, we suggest an approach in the intersection of Social Choice and mechanism design without money. More specifically, following the approach of mechanism design without money, we focus on what we call \emph{nice mechanisms}, namely deterministic strategyproof mechanisms with an approximation ratio bounded by some function of $n$ and $K$, and following the approach of Social Choice, we embark on a characterization of nice mechanisms, instead of seeking stronger lower bounds based on carefully selected instances, as e.g., in \cite{PT09,LWZ09,LSWZ10}.
Although for simplicity and clarity, we focus on the objective of social cost, we highlight that the class of nice mechanisms is very general and essentially independent of the choice of the social objective function. E.g., any mechanism with a bounded approximation ratio for the objective of minimizing the $L_p$ norm of the agents' distances to the nearest facility is nice, for any $p \geq 1$ or for $p = \infty$, since it also achieves a bounded approximation for the objective of social cost. The same holds for the objectives of Sum-$K$-Radii and Sum-$K$-Diameters, and to the best of our knowledge, for any other natural $K$-location objective. Thus, to a very large extent, a characterization of nice mechanisms retains the generality of characterizations in Social Choice, since it captures all, but some socially intolerable, strategyproof mechanisms.

On the other hand, focusing on nice mechanisms facilitates the characterization, since it excludes several socially intolerable strategyproof mechanisms, such as mechanisms with two dictators.
Nevertheless, any characterization of nice mechanisms, even for two facilities, remains an intriguing task, because the preferences are not single-peaked anymore, there is no apparent notion of monotonicity (as e.g., in \cite{Kouts11}), and the combinatorial structure of the problem is significantly more complicated than that for a single facility.

Our main result is an elegant characterization of nice mechanisms for 2-Facility Location on the line. We show that any nice mechanism for $n \geq 5$ agents either admits a unique dictator, or always places the facilities at the two extremes (Theorem~\ref{thm:2fac-gen}). A corollary is that the best approximation ratio achievable by deterministic mechanisms for 2-Facility Location on the line is $n-2$. Another rather surprising consequence is that {\sc Two-Extremes} is the only anonymous nice mechanism for 2-Facility Location on the line.

%At the conceptual level, 
The proof of Theorem~\ref{thm:2fac-gen} proceeds by establishing the characterization at three different levels of generality: 3-agent, 3-location, and general instances. Along the way, we are developing stronger and stronger technical tools that fully describe the behavior of nice mechanisms.
To exploit locality, we first focus on well-separated instances with 3 agents, where an isolated agent is served by one facility, and two nearby agents are served by the other facility. Interestingly, we identify two large classes of well-separated instances where any nice mechanism should keep allocating the latter facility to the same agent (Propositions~\ref{prop:right_cover}~and~\ref{prop:left_cover}).
Building on this, we show that the location of the facility serving the nearby agents is determined by a generalized median voter scheme, as in \cite{Moul80}, but with a threshold depending on the location and the identity of the isolated agent, and then extend this property to general instances with 3 agents (see Fig.~\ref{fig:allocation}). The next key step is to show that the threshold of each isolated agent can only take two extreme values: one corresponding to the existence of a partial dictator, and one corresponding to allocating the facility to the furthest agent (Lemma~\ref{l:i-cases}). Then, considering all possible cases for the agents' thresholds, we show that any nice mechanism for 3 agents either places the facilities at the two extremes, or admits a partial dictator, namely an agent allocated a facility for all, but possibly one, of agent permutations (Theorem~\ref{thm:3agents}).

Next, we employ the notion of partial group strategyproofness \cite[Section~3]{LSWZ10}, and a new technical tool for moving agents between different coalitions without affecting the mechanism's outcome (Lemma~\ref{l:any-partition}), and show that any nice mechanism applied to 3-location instances with $n \geq 5$ agents either admits a (full) dictator, or places the facilities at the two extremes (Theorem~\ref{thm:3locations}). Rather surprisingly, this implies that nice mechanisms for 3 agents are somewhat less restricted than nice mechanisms for $n \geq 5$ agents. Finally, in Section~\ref{s:2fac-gen}, we employ induction on the number of different locations, and conclude the proof of Theorem~\ref{thm:2fac-gen}.

In addition to extending the ideas of \cite{Moul80} to 2-Facility Location games and to exploiting the notions of image sets and partial group strategyproofness, in the proof of Theorem~\ref{thm:2fac-gen}, we introduce a few new ideas and technical tools, which provide new insights into the behavior of nice mechanisms for $K$-Facility Location games, and may be of independent interest. Among them, we may single out the notion of well-separated instances and the idea of reducing $K$-Facility Location in well-separated instances to a single facility game between the two nearby agents, the ideas used to extend the facility allocation from well-separated instances to general instances, the use of thresholds to eliminate non-nice mechanisms, and the technical tool of moving agents between different coalitions without affecting the outcome.

For $K \geq 3$ facilities, we show that there do not exist any deterministic anonymous strategyproof mechanisms with a bounded approximation ratio, which holds even for well-separated instances with $K+1$ agents on the line (Theorem~\ref{thm:3facilities}). For 2-Facility Location in metric spaces more general than the line and the circle, we show that there do not exist any deterministic strategyproof mechanisms with a bounded approximation ratio, which holds even for simple instances with $3$ agents located in a star (Theorem~\ref{thm:tree-unbounded}). Both results are based on the technical tools for well-separated instances developed in the proof of Theorem~\ref{thm:2fac-gen}, thus indicating  the generality and the potential applicability of our techniques. At the conceptual level, the proof approach of Theorem~\ref{thm:2fac-gen} and the proofs of Theorem~\ref{thm:3facilities} and Theorem~\ref{thm:tree-unbounded} imply that the instances with $K+1$ agents are among the hardest ones for deterministic $K$-Facility Location mechanisms.

\smallskip\noindent{\bf Other Related Work.}
In Social Choice, the work on multiple facility location games mostly focuses on Pareto-efficient strategyproof mechanisms that satisfy replacement-domination \cite{Miy01}, and on Pareto-efficient mechanisms whose outcome is consistent with the decisions of the agents served by the same facility \cite{BB06,Ju08}. However, these conditions do not have any immediate implications for the approximability of the social cost (or of any other social objective), and thus, we cannot technically exploit these results.

\smallskip\noindent{\em Locating a Single Facility.}
Alon at al. \cite{AFPT09} almost completely characterized the approximation ratios achievable by randomized and deterministic mechanisms for 1-Facility Location in general metrics and rings.
Next, Feldman and Wilf \cite{FW11} proved that for the $L_2$ norm of the distances to the agents, the best approximation ratio is $1.5$ for randomized and $2$ for deterministic mechanisms. Moreover, they presented a class of randomized mechanisms that includes all known strategyproof mechanisms for $1$-Facility Location on the line.

\smallskip\noindent{\em Locating Multiple Facilities.}
For $K \geq 4$, the case of $K+1$ agents is the only case where a (randomized) strategyproof mechanism with a bounded approximation ratio is known. Escoffier at al. \cite{EGTPS11} proved that in this case, the {\sc Inversely Proportional} mechanism is strategyproof and achieves an approximation ratio of $(K+1)/2$ for $K$-Facility Location in general metric spaces. Interestingly, Theorem~\ref{thm:3facilities} shows that these instances are among the hardest ones for deterministic anonymous mechanisms.

\smallskip\noindent{\em Imposing Mechanisms.}
Nissim at al. \cite{NST10} introduced the notion of \emph{imposing mechanisms}, where the mechanism can restrict how agents exploit its outcome, and thus increase their individual cost if they lie (e.g., for Facility Location games, an imposing mechanism can forbid an agent to connect to some facilities).
They combined the almost-strategyproof differentially private mechanism of \cite{MT07} with an imposing mechanism that penalizes lying agents, and obtained a general randomized imposing strategyproof mechanism. As a by-product, they obtained a randomized imposing mechanism for $K$-Facility Location that approximates the average optimal social cost within an additive term of roughly $1/n^{1/3}$.
Subsequently, we proved, in \cite{FT10}, that the imposing version of the {\sc Proportional} mechanism is strategyproof for $K$-Facility Location in general metric spaces, and achieves an approximation ratio of at most $4K$.

\section{Notation, Definitions, and Preliminaries}
\label{s:prelim}

With the exception of Section~\ref{s:tree}, we consider $K$-Facility Location on the real line. So, in this section, we introduce the notation and the basic notions only for instances on the real line. Throughout this work, we let $K$-Facility Location refer to the problem of placing $K$ facilities on the  real line, unless stated otherwise.

\smallskip\noindent{\bf Notation.}
For a tuple $\vec{x} = (x_1, \ldots, x_n) \in \reals^n$, $\min\vec{x}$, $\max\vec{x}$, and $\med\vec{x}$ denote the smallest, the largest, and the $\tceil{n/2}$-smallest coordinate of $\vec{x}$, respectively.
We let $\vec{x}_{-i}$ %= (x_1, \ldots, x_{i-1}, x_{i+1}, \ldots, x_n)$
be the tuple $\vec{x}$ without $x_i$. For a non-empty set $S$ of indices, we let $\vec{x}_S = (x_i)_{i \in S}$ and $\vec{x}_{-S} = (x_i)_{i \not\in S}$.
We write %$\vec{x} = (\vec{x}_{-i}, x_i)$ and $\vec{x} = (\vec{x}_{-S}, \vec{x}_S)$.
$(\vec{x}_{-i}, a)$ to denote the tuple $\vec{x}$ with $a$ in place of $x_i$, $(\vec{x}_{-\{i,j\}}, a, b)$ to denote the tuple $\vec{x}$ with $a$ in place of $x_i$ and $b$ in place of $x_j$, and so on.

\smallskip\noindent{\bf Instances.}
Let $N = \{ 1, \ldots, n\}$ be a set of $n \geq 3$ agents. Each agent $i \in N$ has a location $x_i \in \reals$, which is $i$'s private information. We usually refer to a locations profile $\vec{x} = (x_1, \ldots, x_n) \in \reals^n$ as an \emph{instance}.
For an instance $\vec{x}$, we say that the agents are arranged on the line according to a permutation $\pi$ if $\pi$ arranges them in increasing order of their locations in $\vec{x}$, i.e., $x_{\pi(1)} \leq  x_{\pi(2)} \leq \cdots \leq x_{\pi(n)}$.
In the proof of Theorem~\ref{thm:2fac-gen}, we consider \emph{3-agent instances}, where $n = 3$, and \emph{3-location instances}, where there are three different locations $x_1, x_2, x_3$, and a partition of $N$ into three coalitions $N_1, N_2, N_3$ such that all agents in coalition $N_i$ occupy location $x_i$, $i \in \{1, 2, 3\}$.
We usually denote such an instance as $(x_1\sep N_1, x_2\sep N_2, x_3\sep N_3)$.
For a set $N$ of agents, we let $\I(N)$ denote the set of all instances, and let $\Ithr(N)$ denote the set of all 3-location instances.

\smallskip\noindent{\bf Mechanisms.}
A (deterministic) mechanism $F$ for $K$-Facility Location maps an instance $\vec{x}$ to a $K$-tuple $(y_1, \ldots, y_K) \in \reals^K$, $y_1 \leq \cdots \leq y_K$, of facility locations.
We let $F(\vec{x})$ denote the outcome of $F$ for instance $\vec{x}$, and let $F_\ell(\vec{x})$ denote $y_\ell$, i.e., the $\ell$-th smallest coordinate in $F(\vec{x})$. In particular, for 2-Facility Location, $F_1(\vec{x})$ denotes the leftmost and $F_2(\vec{x})$ denotes the rightmost facility of $F(\vec{x})$. %Slightly abusing the notation,
We write $y \in F(\vec{x})$ to denote that $F(\vec{x})$ has a facility at $y$.
A mechanism $F$ is \emph{anonymous} if for all instances $\vec{x}$ and all agent permutations $\pi$, $F(\vec{x}) = F(x_{\pi(1)}, \ldots, x_{\pi(n)})$.
Throughout this work, all references to a mechanism $F$ assume a deterministic mechanism, unless explicitly stated otherwise.

\smallskip\noindent{\bf Social Cost.}
Given a mechanism $F$ for $K$-Facility Location and an instance $\vec{x}$, the (individual) cost of agent $i$ is
\( \cost[x_i, F(\vec{x})] = \min_{1 \leq \ell \leq K}\{ |x_i - F_\ell(\vec{x})| \} \).
The (social) cost of $F$ for an instance $\vec{x}$ is
\( \cost[F(\vec{x})] = \sum_{i = 1}^n \cost[x_i, F(\vec{x})] \).
The optimal cost for an instance $\vec{x}$ is $\min \sum_{i=1}^n \cost[x_i, (y_1, \ldots, y_K)]$, where the minimum is taken over all $K$-tuples $(y_1, \ldots, y_K)$. %\in \reals^K$.

A mechanism $F$ has an approximation ratio of $\rho \geq 1$, if for any instance $\vec{x}$, the cost of $F(\vec{x})$ is at most $\rho$ times the optimal cost for $\vec{x}$.
We say that the approximation ratio $\rho$ of $F$ is \emph{bounded} if $\rho$ is either some constant or some (computable) function of $n$ and $K$.
%
%(but it does not depend on other parameters of the instance, such as the maximum or the minimum distance between any pair of agents).
%
Since for any $p \geq 1$ (or for $p = \infty$), and for any non-negative $n$-tuple $\vec{c}$, $\| c \|_p \leq \sum_{i=1}^n c_i \leq n^{1-1/p} \| c \|_p$, a mechanism with a bounded approximation ratio for the $L_p$ norm of the agents' individual costs also has a bounded approximation ratio for the social cost.

\smallskip\noindent{\bf Strategyproofness.}
%and (Partial) Group Strategyproofness.}
%
A mechanism $F$ is \emph{strategyproof} if no agent can benefit from misreporting her location. Formally, for all instances $\vec{x}$, every agent $i$, and all locations $y$, %it holds that
\( \cost[x_i, F(\vec{x})] \leq \cost[x_i, F(\vec{x}_{-i}, y)] \).
A mechanism $F$ is \emph{group strategyproof} if for any coalition of agents misreporting their locations, at least one of them does not benefit. Formally,
for all instances $\vec{x}$, every coalition of agents $S$, and all subinstances $\vec{y}_S$, there exists some agent $i \in S$ such that
\( \cost[x_i, F(\vec{x})] \leq \cost[x_i, F(\vec{x}_{-S},\vec{y}_S)] \).
A mechanism $F$ is \emph{partial group strategyproof} if for any coalition of agents that occupy the same location, none of them can benefit if they misreport their location simultaneously. Formally, for all instances $\vec{x}$, every coalition of agents $S$, all occupying the same location $x$ in $\vec{x}$, and all subinstances $\vec{y}_S$, %it holds that
\( \cost[x, F(\vec{x})] \leq \cost[x, F(\vec{x}_{-S}, \vec{y}_S)] \).

By definition, any group strategyproof mechanism is partial group strategyproof, and any partial group strategyproof mechanism is strategyproof.
In \cite[Lemma~2.1]{LSWZ10}, it is shown that any strategyproof mechanism for $K$-Facility Location is also partial group strategyproof (see also \cite[Section~2]{Moul80}).

\smallskip\noindent{\bf Image Sets.}
Given a mechanism $F$, the \emph{image} (or option) \emph{set} $I_i(\vec{x}_{-i})$ of an agent $i$ with respect to an instance $\vec{x}_{-i}$ is the set of facility
locations the agent $i$ can obtain by varying her reported location. Formally,
\( I_i(\vec{x}_{-i}) = \{ a \in \reals: \exists y \in \reals \mbox{ such that } F(\vec{x}_{-i}, y) = a \} \).
%
%One can show that
If $F$ is strategyproof, any image set $I_i(\vec{x}_{-i})$ is a collection of closed intervals, and $F$ places a facility at the location in $I_i(\vec{x}_{-i})$ nearest to the location of agent $i$. Formally, for any agent $i$, all instances $\vec{x}$, and all locations $y$, %it holds that
\( \cost[y, F(\vec{x}_{-i}, y)] =
\inf_{a \in I_i(\vec{x}_{-i})}\{|y - a|\} \).
In \cite[Section~3.1]{LSWZ10}, it is shown that using partial group strategyproofness, we can extend the notion of image sets and the properties above to coalitions of agents that occupy the same location in an instance $\vec{x}$.

Any (open) interval in the complement of an image set $I \equiv I_i(\vec{x}_{-i})$ is called a \emph{hole} of $I$. Given a location $y \not\in I$, we let $l_y = \sup_{a \in I} \{ a < y\}$ and $r_y = \inf_{a \in I} \{ a > y\}$ be the locations in $I$ nearest to $y$ on the left and on the right, respectively. Since $I$ is a collection of closed intervals, $l_y$ and $r_y$ are well defined and satisfy $l_y < y < r_y$. For convenience, given a $y \not\in I$, we refer to the interval $(l_y, r_y)$ as a $y$-hole in $I$.

\smallskip\noindent{\bf Nice Mechanisms.}
For simplicity, we use the term \emph{nice mechanism} to refer to any mechanism $F$ that is deterministic, strategyproof, and has a bounded approximation ratio. We usually refer to a nice mechanism $F$ without explicitly mentioning its approximation ratio, with the understanding that given $F$ and the set $N$ of agents, we can determine an upper bound $\rho$ on the approximation ratio of $F$ for instances in $\I(N)$.

Any nice mechanism $F$ for $K$-Facility Location is \emph{unanimous}, namely for all instances $\vec{x}$ where the agents occupy $K$ different locations $x_1, \ldots, x_K$, $F(\vec{x}) = (x_1, \ldots, x_K)$. %since otherwise $F$ would not have a bounded approximation ratio.
Similarly, any hole in an image set $I_i(\vec{x}_{-i})$ of $F$ is a bounded interval. Otherwise, %i.e., if there was an image set $I_i(\vec{x}_{-i})$ with a hole that extends either to $-\infty$ or to $+\infty$, then
we could move agent $i$ sufficiently far away from the remaining agents, and obtain an instance for which $F$ would have approximation ratio larger than $\rho$. %Therefore, if $F$ is a nice mechanism, for any instance $\vec{x}$ and any agents $i$, there is a sufficiently small (resp. large) $a$, such that if $i$ moves to $a$, $F$ allocates a facility to $a$, i.e, $a \in F(\vec{x}_{-i}, a)$.

\smallskip\noindent{\bf Well-Separated Instances.}
Given a nice mechanism $F$ for $K$-Facility Location with approximation ratio $\rho$, a $(K+1)$-agent instance $\vec{x}$ is called \emph{$(i_1|\cdots|i_{K-1}|i_K, i_{K+1})$-well-separated} if $x_{i_1} < \cdots < x_{i_{K+1}}$ and
\( \rho(x_{i_{K+1}} - x_{i_K}) < \min_{2 \leq \ell \leq K} \{ x_{i_\ell} - x_{i_{\ell-1}} \} \).
At the conceptual level, in a well-separated instance, there is a pair of nearby agents whose distance to each other is less than $1/\rho$ times the distance between any other pair of consecutive agent locations on the real line. Therefore any mechanism with an approximation ratio of at most $\rho$ should serve the two nearby agents by the same facility, and serve each of the remaining ``isolated'' agents by a different facility.

\subsection{Useful Properties}
\label{s:properties}

%\smallskip\noindent{\bf Properties.}
%
We present here some useful properties of nice mechanisms for $K$-Facility Location on the real line applied to instances with $K+1$ agents.

\begin{proposition}\label{prop:middle}
Let $F$ be a nice mechanism for $K$-Facility Location on the line. For any $(K+1)$-location instance $\vec{x}$ with $x_{i_1} < \cdots < x_{i_{K+1}}$, $F_1(\vec{x}) \leq x_{i_2}$ and $F_K(\vec{x}) \geq x_{i_K}$.
\end{proposition}

\begin{proof}
Let us assume that $x_{i_2} < F_1(\vec{x})$ (the other case is symmetric). Then, the agents at $x_{i_1}$ have an incentive to report $x_{i_2}$ and decrease their cost, since $x_{i_2} \in F(\vec{x}_{-i_1}, x_{i_2})$, due to the bounded approximation ratio of $F$. This contradicts $F$'s (partial group) strategyproofness.
\qed\end{proof}

\begin{proposition}\label{prop:interval}
Let $F$ be a nice mechanism for $K$-Facility Location. For any $(i_1|\cdots|i_{K-1}|i_K, i_{K+1})$-well-separated instance $\vec{x}$, $F_K(\vec{x}) \in [x_{i_K}, x_{i_{K+1}}]$.
\end{proposition}

\begin{proof}
Since $F$ has a bounded approximation ratio, the two nearby agents $i_K$ and $i_{K+1}$ are both served by the facility at $F_K(\vec{x})$. By Proposition~\ref{prop:middle}, $F_K(\vec{x}) \geq x_{i_K}$. Moreover, $F_K(\vec{x}) \leq x_{i_{K+1}}$. Otherwise, the agent $i_K$ could report $x_{i_{K+1}}$ and decrease her cost, since $x_{i_{K+1}} \in F(\vec{x}_{-i_K}, x_{i_{K+1}})$, due to the bounded approximation ratio of $F$.
\qed\end{proof}

The following propositions show that if there exists an $(i_1|\cdots|i_{K-1}|i_K, i_{K+1})$-well-separated instance $\vec{x}$ with $F_K(\vec{x}) = x_{i_K}$ (resp. $F_K(\vec{x}) = x_{i_{K+1}}$), then as long as we ``push'' the locations of agents $i_K$ and $i_{K+1}$ to the right (resp. left), while keeping the instance well-separated, the rightmost facility of $F$ stays with the location of agent $i_K$ (resp. $i_{K+1}$). The proofs can be found in the Appendix, Section~\ref{app:s:push_right} and Section~\ref{app:s:push_left}, respectively.
We should highlight that one can establish the equivalent of the following propositions for well-separated instances where the two nearby agents are located elsewhere in the instance (e.g., the nearby agents are the two leftmost agents, or the second and third agent from the left).

\begin{proposition}\label{prop:right_cover}
Let $F$ be a nice mechanism for $K$-Facility Location, and $\vec{x}$ be a $(i_1|\cdots|i_{K-1}|i_K, i_{K+1})$-well-separated instance with $F_K(\vec{x}) = x_{i_K}$. Then for every $(i_1|\cdots|i_{K-1}|i_K, i_{K+1})$-well-separated instance $\vec{x}' = (\vec{x}_{-\{i_K, i_{K+1}\}}, x'_{i_K}, x'_{i_{K+1}})$ with $x_{i_K} \leq x'_{i_K}$, it holds that $F_K(\vec{x}') = x'_{i_K}$\,.
\end{proposition}

\begin{proposition}\label{prop:left_cover}
Let $F$ be a nice mechanism for $K$-Facility Location, and $\vec{x}$ be a $(i_1|\cdots|i_{K-1}|i_K, i_{K+1})$-well-separated instance with $F_K(\vec{x}) = x_{i_{K+1}}$. Then for every $(i_1|\cdots|i_{K-1}|i_K, i_{K+1})$-well-separated instance $\vec{x}' = (\vec{x}_{-\{i_K, i_{K+1}\}}, x'_{i_K}, x'_{i_{K+1}})$ with $x'_{i_{K+1}} \leq x_{i_{K+1}}$, it holds that $F_K(\vec{x}') = x'_{i_{K+1}}$\,.
\end{proposition} 

\section{Strategyproof Mechanisms for 2-Facility Location: Outline}
\label{s:2facilities}

%We proceed to consider the case where the mechanism allocates 2 facilities. %
We proceed to discuss the key proof steps and the consequences of our main result:
%namely that any nice mechanism for 2-Facility Location, applied to instances with $n \geq 5$ agents, either admits a dictator, or always allocates the facilities to the two extreme locations of the instance. More formally, we show the following:

\begin{theorem}\label{thm:2fac-gen}
Let $F$ be a nice mechanism for 2-Facility Location with $n \geq 5$ agents. Then, either $F(\vec{x}) = ( \min \vec{x}, \max \vec{x} )$ for all instances $\vec{x}$, or there exists a unique dictator $j$ such that for all $\vec{x}$, $x_j \in F(\vec{x})$.
\end{theorem}

Notably, we are aware of only two nice mechanisms for 2-Facility Location with $n \geq 4$ agents, one for each case of Theorem~\ref{thm:2fac-gen}.
The {\sc Dictatorial} mechanism chooses a dictator $j$, and for each instance $\vec{x}$, allocates a facility to $x_j$. Then, it considers the distance of the dictator to the leftmost and to the rightmost location, $d_l = |\min\vec{x} - x_j|$ and $d_r = |\max\vec{x} - x_j|$, respectively. The second facility is placed at $x_j - \max \{ d_l, 2 d_r \}$, if $d_l > d_r$, and to $x_j + \max \{ d_r, 2 d_l \}$, otherwise.
{\sc Dictatorial} is an adaptation of the mechanism of \cite{LSWZ10} for the circle. As in \cite[Section~5]{LSWZ10}, it can be shown that {\sc Dictatorial} is strategyproof and $(n-1)$-approximate for the line.
The {\sc Two-Extremes} mechanism places the facilities at $(\min\vec{x}, \max\vec{x})$, for all instances $\vec{x}$, and is group strategyproof, anonymous, and $(n-2)$-approximate \cite{PT09}.
By Theorem~\ref{thm:2fac-gen}, {\sc Two-Extremes} is the only anonymous nice mechanism for 2-Facility Location with $n \geq 5$ agents and its approximation ratio is best possible. 

\begin{corollary}\label{cor:2fac-anonymous}
A deterministic anonymous mechanism $F$ for 2-Facility Location with $n \geq 5$ agents is strategyproof and has a bounded approximation ratio if and only if for all instances $\vec{x}$, $F(\vec{x}) = ( \min \vec{x}, \max \vec{x} )$.
\end{corollary}

\begin{corollary}\label{cor:2fac-approx}
Any deterministic strategyproof mechanism $F$ for 2-Facility Location with $n \geq 5$ agents has an approximation ratio of at least $n-2$.
\end{corollary}

\begin{proof}
For any set $N$ of $n \geq 5$ agents, we let $\vec{x} = (0\sep \{j\}, \eps\sep N\setminus\{j, k\}, 1\sep \{k\})$, where $j \in N$ is the dictator of $F$, if any, $k \in N \setminus \{j\}$, and $\eps \in (0, 1/n)$. By Theorem~\ref{thm:2fac-gen}, if $F$ has a bounded approximation ratio, then $F_1(\vec{x}) = 0$. For convenience, we let $F_2(\vec{x}) = a$.
By Proposition~\ref{prop:middle}, $a \geq \eps$. If $\eps \leq a < 2\eps$, $\cost[F(\vec{x})] \geq 1-\eps$. If $a \geq 2\eps$, $\cost[F(\vec{x})] \geq (n-2)\eps$. Since the optimal cost is $\eps$, the approximation of $F$ is at least $n-2$. 
\qed\end{proof}

The crux, and the most technically involved part of the proof of Theorem~\ref{thm:2fac-gen} is to establish a characterization for nice mechanisms dealing with just 3 agents. In particular, we show that any nice mechanism for 2-Facility Location with 3 agents either places the facilities at the two extremes, or admits a \emph{partial dictator}, namely an agent allocated a facility either for all agent permutations or for all agent permutations but one.

\begin{theorem}\label{thm:3agents}
Let $F$ be any nice mechanism for 2-Facility Location with $n = 3$ agents. Then, there exist at most two permutations $\pi_1$, $\pi_2$ with $\pi_1(2) = \pi_2(2)$ such that for all instances $\vec{x}$ where the agents are arranged on the line according to $\pi_1$ or $\pi_2$, $\med \vec{x} \in F(\vec{x})$. For any other permutation $\pi$ and instance $\vec{x}$, where the agents are arranged on the line according to $\pi$, $F(\vec{x}) = ( \min \vec{x}, \max \vec{x} )$.
\end{theorem}

In addition to extending the ideas of \cite{Moul80} to 2-Facility Location games, the proof of Theorem~\ref{thm:3agents} makes, at several places, a novel use of locality and of the structure of the image sets of nice mechanisms.
Also, we highlight that the notion of a partial dictator is essential. The {\sc Combined} mechanism for 3 agents chooses a permutation $(i, j, k)$ of the agents, and for each instance $\vec{x}$, places the facilities using {\sc Two-Extremes}, if $x_i < x_k$, and using {\sc Dictatorial} with dictator $j$, otherwise. Thus, {\sc Combined} admits a partial dictator, is strategyproof, and achieves an approximation ratio of 2.

Using the notion of partial group strategyproofness, we extend Theorem~\ref{thm:3agents} to 3-location instances (Corollary~\ref{cor:3agents}).
The next step is to show that when applied to 3-location instances with $n \geq 5$ agents, nice mechanisms do not have the option of a partial dictator. More formally, in Section~\ref{s:3locations}, we show the following:

\begin{theorem}\label{thm:3locations}
Let $N$ be a set of $n \geq 5$ agents, and let $F$ be any nice mechanism for 2-Facility Location applied to instances in $\Ithr(N)$. Then, either there exists a unique dictator $j \in N$ such that for all instances $\vec{x} \in \Ithr(N)$, $x_j \in F(\vec{x})$, or for all instances $\vec{x} \in \Ithr(N)$, $F(\vec{x}) = ( \min \vec{x}, \max \vec{x} )$.
\end{theorem}

Finally, in Section~\ref{s:2fac-gen}, we employ induction on the number of agents, and extend Theorem~\ref{thm:3locations} to general instances with $n \geq 5$ agents, thus concluding the proof of Theorem~\ref{thm:2fac-gen}.
In the next three sections, we present a detailed proof of Theorem~\ref{thm:2fac-gen}. Throughout, we usually omit any quantification of $F$, with the understanding that $F$ denotes a nice mechanism for 2-Facility Location applied to the relevant class of instances.

\section{Strategyproof Allocation of 2 Facilities to 3 Agents: The Proof of Theorem~\ref{thm:3agents}}
\label{s:3agents}

Throughout this section, we use the indices $i, j, k$ to implicitly define a permutation of the agents. We mostly use the convention that $i$ denotes the leftmost agent, $j$ denotes the middle agent, and $k$ denotes the rightmost agent. %If a few different instances are involved, this convention applies to the first instance.

We recall that given a nice mechanism $F$ with approximation ratio $\rho$ for $3$ agents, a $3$-agent instance $\vec{x}$ is \emph{$(i | j, k)$-well-separated} if $x_i < x_j < x_k$ and $\rho(x_k - x_j) < x_j - x_i$. Similarly, $\vec{x}$ is \emph{$(i, j | k)$-well-separated} if $x_i < x_j < x_k$ and $\rho(x_j - x_i) < x_k - x_j$.
A $3$-agent instance $\vec{x}$ is \emph{$i$-left-well-separated} if $\vec{x}$ is either $(i|j,k)$-well-separated or $(i|k,j)$-well-separated, and is \emph{$k$-right-well-separated} if it is either $(i,j|k)$-well-separated or $(j,i|k)$-well-separated. Moreover, a $3$-agent instance $\vec{x}$ is \emph{$i$-well-separated} if $\vec{x}$ is either $i$-left-well-separated or $i$-right-well-separated.

In the following, we let $\uparrow$ and $\downarrow$ denote the largest and the smallest element, respectively, of the affinely extended real line. Hence, $\uparrow$ is greater than any real number, and $\downarrow$ is less than any real number.

\subsection{Outline of the Proof}
\label{s:outline}

At a high-level, the proof of Theorem~\ref{thm:3agents} proceeds by gradually restricting the possible outcomes of a nice mechanism, until it reaches the desired conclusion.

As a first step, we consider the behavior of nice mechanisms for well-separated instances (Section~\ref{s:well-separated}). Since the mechanism has a bounded approximation ratio, for any $i$-well-separated instance, one facility serves the isolated agent $i$, and the other facility is placed between the locations of the two nearby agents $j$, $k$ (Proposition~\ref{prop:interval}). Thus, building on the characterization of \cite{Moul80}, we show that for any $i$-well-separated instance, the facility serving agents $j$ and $k$ is allocated by a generalized median voter scheme (GMVS) (see e.g., \cite[Definition~10.3]{SV07}) whose characteristic threshold may depend on the identity $i$ and the location $a$ of the isolated agent (Lemma~\ref{l:single-p} and Lemma~\ref{l:i-separated}). A bit more formally, we show, in Lemma~\ref{l:i-separated}, that any agent-location pair $(i, a)$ specifies a unique threshold $p \in [a, +\infty) \union \{ \uparrow \}$ and a preferred agent $j$, different from $i$, that fully determine the location of the rightmost facility for all $i$-left-well-separated instances $\vec{x}$ with $x_i = a$ (see also Fig.~\ref{fig:allocation}.a; by symmetry, the same holds for $i$-right-well-separated instances and the leftmost facility, though possibly with different values of $p$ and $j$). Moreover, the allocation of the rightmost facility becomes simple for the two extreme values of the $p$: if $p = a$, the preferred agent $j$ serves as a dictator imposed by $i$ for all $i$-left-well-separated instances, while if $p =\,\uparrow$, the rightmost facility is placed at the rightmost location.

%%%%%%%%%%%%%%%%%%%%%%%%%%%%%%%%%%%%%%%%%%%%%%%%%%%%%%%%%%%%%%%%%%%%
\insertfig{0.85\textwidth}{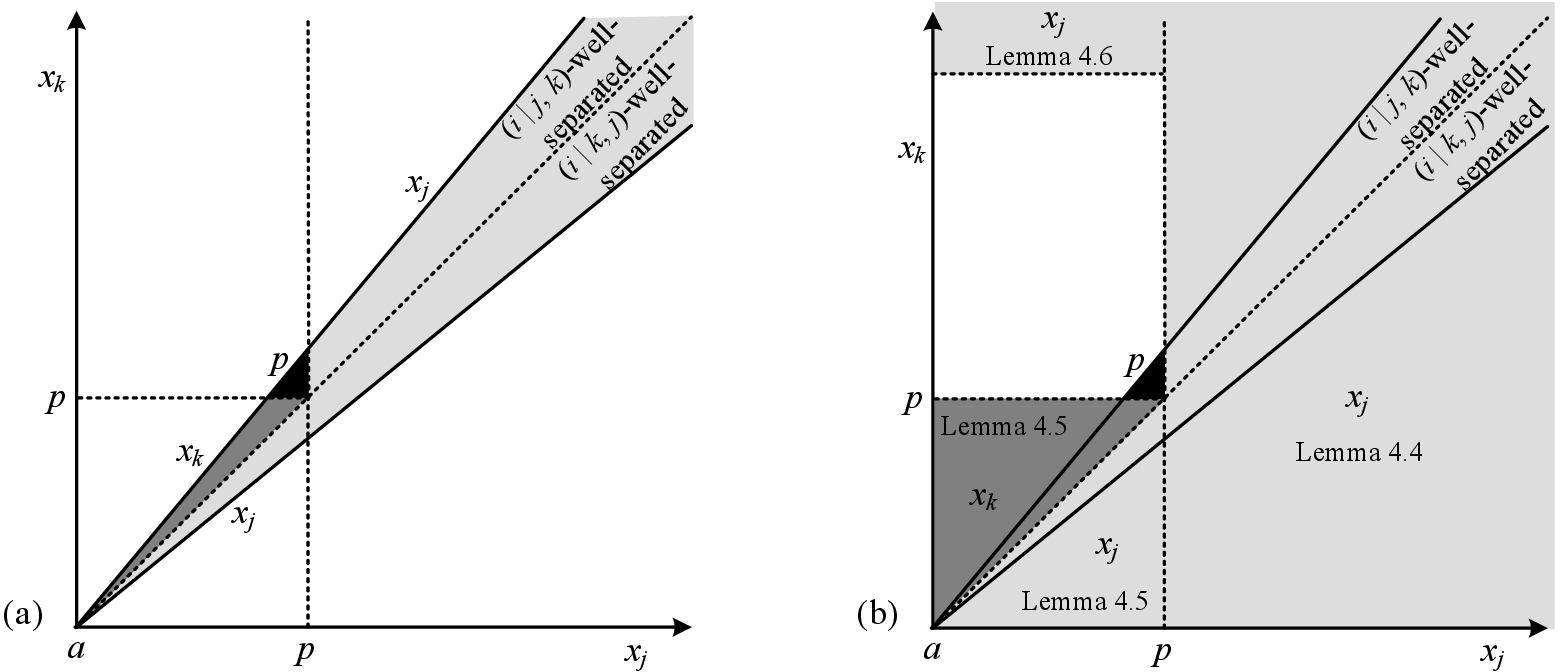}{\label{fig:allocation}%
(a) The location of $F_2(\vec{x})$ for all $i$-left-well-separated instances $\vec{x}$ with $x_i = a$. We let $j$ be the preferred agent and $p$ be the threshold of $(i, a)$ for $i$-left-well-separated instances. The location of agent $j$ (resp. $k$) is on the $x$-axis (resp. $y$-axis). The area around the line $x_j = x_k$ includes all $i$-left-well-separated instances. For instances in the dark grey area (where $x_j \leq x_k \leq p$), the rightmost facility is at $x_k$, for instances in the black triangle (where $x_j \leq p \leq x_k$), the rightmost facility is at $p$, and for instances in the light grey area (where either $x_j \geq p$ or $x_k \leq x_j \leq p$), the rightmost facility is at $x_j$. %We highlight that the smaller the value of $p$, the larger the area where $F_2(\vec{x}) = x_j$.
(b) We consider instances $\vec{x}$ with $x_i = a < x_j, x_k$, which are not necessarily well-separated. As in (a), we let $j$ be the preferred agent and $p$ be the threshold of $(i, a)$.
%The location of agent $j$ (resp. $k$) is on the $x$-axis (resp. $y$-axis).
The plot now depicts for which instances a facility (not necessarily the rightmost one) is placed at either $x_j$, or $x_k$, or $p$.
%
%By Lemma~\ref{l:non-separated}, for all instances $\vec{x}$ with $x_j \geq p$ (the light grey area on the right), $x_j \in F(\vec{x})$. By Lemma~\ref{l:non-separated-inside}, for all instances $\vec{x}$ with $a < x_j, x_k \leq p$, $F_2(\vec{x}) = \max\{x_j, x_k\}$. Therefore, for instances in the (dark grey) upper-left part of the square $a < x_j, x_k \leq p$, the rightmost facility is allocated to $x_k$, while for instances in the (light grey) lower-right part, the rightmost facility is allocated to $x_j$. Finally, by Lemma~\ref{l:approx}, for all instances $\vec{x}$ with $a < x_j < x_k$ and $x_k$ large enough (the light grey area on the top-left corner), $x_j \in F(\vec{x})$.
%
Lemma~\ref{l:i-cases} shows that due to the bounded approximation ratio of $F$, the black triangle does not exist, and either $p = a$ or $p =\,\uparrow$. If $p = a$, we have light grey (i.e., $x_j$) everywhere, while if $p = \,\uparrow$, we have dark grey (i.e., $x_k$) above the line $x_j = x_k$, and light grey (i.e., $x_j$) below.
}
%%%%%%%%%%%%%%%%%%%%%%%%%%%%%%%%%%%%%%%%%%%%%%%%%%%%%%%%%%%%%%%%%%%%

The key step in the proof is to show, in Section~\ref{s:gen-instances}, that the threshold $p$ of any agent-location pair $(i, a)$ can be either $a$ or $\uparrow$ (resp. either $a$ or $\downarrow$ if $i$ is the rightmost agent).
To this end, we first extend the allocation of Lemma~\ref{l:i-separated} to general instances with $i$ as the leftmost (resp. rightmost) agent (see Fig.~\ref{fig:allocation}.b).
%
%In a nutshell, we show that for most such instances, the threshold $p$ and the preferred agent $j$, that determine the location of a facility on the right (resp. left) part of the instance, %may depend on the identity $i$ and the location $a$ of the leftmost (resp. rightmost) agent, and
%are the same as those for $i$-left (resp. $i$-right) well-separated instances.
Thus, if the preferred agent $j$ is located on the right (resp. left) of the threshold $p$, she essentially serves as a partial dictator, imposed by the leftmost (resp. rightmost) agent, for the corresponding permutation of agents.

As consequence, we obtain that the thresholds of the two allocation rules (one imposed by the leftmost agent and one imposed by the rightmost agent) always fall in the two extremes: either $a$ or $\uparrow$ for the leftmost agent (resp. either $a$ or $\downarrow$ for the rightmost agent). Otherwise, there would exist instances with two different preferred agents, essentially serving as two different partial dictators, which would lead to an unbounded approximation ratio (Lemma~\ref{l:i-cases}). Intuitively, the case where $p = a$ corresponds to the existence of a partial dictator, while the case where $p =\,\uparrow$ (resp. $p =\,\downarrow$) corresponds to placing the facilities at the two extremes.

%Building on Lemma~\ref{l:i-cases}, we show, in Section~\ref{s:3agent-fin}, that every nice mechanism is essentially characterized by whether there are two agents that agree on imposing the third agent as a partial dictator or not. By examining all possible cases, we conclude that every nice mechanism $F$ either always places the facilities at the two extremes, or admits a partial dictator $j$ (Lemma~\ref{l:3agents}).

Building on Lemma~\ref{l:i-cases}, we show, in Section~\ref{s:3agent-fin}, that the thresholds of the two allocation rules, one for the leftmost and one for the rightmost agent, can only depend on their identity, and not on their location (Lemma~\ref{l:always-a}). Moreover, if an agent $i$ imposes a partial dictator, the third agent agrees with $i$ not only on the existence of a partial dictator, but also on the dictator's identity (Lemma~\ref{l:others-a}), and the partial dictator is unique (Lemma~\ref{l:only-dictator}).
Therefore, every nice mechanism is essentially characterized by whether there are two agents that agree on imposing the third agent as a partial dictator or not. Examining all possible cases, we conclude that every nice mechanism $F$ either always places the facilities at the two extremes, or admits a partial dictator $j$ (Lemma~\ref{l:3agents}).
In the latter case, the partial dictator $j$ is identified by any instance $\vec{x}$, with $x_i < x_j < x_k$, such that $x_j \in F(\vec{x})$. Then, $F$ allocates a facility to agent $j$ for all instances $\vec{y}$ with $y_i < y_j < y_k$, and possibly for all instances $\vec{y}$ with $y_k < y_j < y_i$. For all other instances, $F$ places the facilities at two extremes.

\subsection{Well-Separated Instances}
\label{s:well-separated}

We proceed to a detailed consideration of the behavior of a nice mechanism on well-separated instances.
For simplicity and brevity, we mostly discuss left-well-separated instances here, for which we state and prove all our technical claims. It is straightforward to verify that the symmetric statements of all our lemmas, propositions, and technical arguments hold for right-well-separated instances.

For any $(i | j, k)$-well-separated instance $\vec{x}$, the leftmost facility $F_1(\vec{x})$ of a nice mechanism $F$ serves the isolated agent $i$, and the rightmost facility $F_2(\vec{x})$ is allocated in $[x_j, x_k]$.
Intuitively, with both the order of $j$ and $k$ and the range of $F_2(\vec{x})$ fixed, the restriction of $F$ to the $2$-agent subinstance $\vec{x}_{-i}$ should behave like an anonymous strategyproof mechanism that places a single facility in $[x_j, x_k]$. Therefore, the rightmost facility should be allocated by a median voter scheme applied to $x_j$ and $x_k$ (see e.g. \cite[Theorem~10.2]{SV07}). The following lemma %, whose proof can be found in the Appendix, Section~\ref{app:s:single-p},
formalizes this intuition:

 %it shows that $F_2(\vec{x}) = \med(p, x_j, x_k)$, where $p$ may depend on the identity $i$ and the location $a$ of the isolated agent, and also on the relative order of the nearby agents $j$ and $k$.

\begin{lemma}\label{l:single-p}
For any agent $i \in \{1, 2, 3\}$ and any location $a$, there exists a unique threshold $p \in [a, +\infty) \union \{ \uparrow \}$ such that for all $(i | j, k)$-well-separated instances $\vec{x}$ with $x_i = a$, it holds that $F_2(\vec{x}) = \med(p, x_j, x_k)$.
\end{lemma}

\begin{proof}
Let $i \in \{ 1, 2, 3\}$ be any agent and $a \in \reals$ be any location. By Proposition~\ref{prop:interval}, for every $(i|j,k)$-well-separated instance $\vec{z}$ with $z_i = a$, $F_2(\vec{z}) \in [z_j, z_k]$. We distinguish between the case where there is an $(i|j,k)$-well-separated instance $\vec{z}$ with $z_i = a$ and $F_2(\vec{z}) \in (z_j, z_k)$, and the case where for all $(i|j,k)$-well-separated instances $\vec{z}$ with $z_i = a$, $F_2(\vec{z}) \in \{ z_j, z_k \}$.

\smallskip\noindent {\bf Case I: $\exists \vec{z}\,F_2(\vec{z}) \in (z_j, z_k)$.}
Let $\vec{z}$ be any $(i|j,k)$-well-separated instance with $z_i = a$ and $F_2(\vec{z}) \in (z_j, z_k)$. In this case, we let $p = F_2(\vec{z})$. Next, we show that for any $(i | j, k)$-well-separated instance $\vec{x}$ with $x_i = a$\,:
\begin{equation}\label{eq:p-median}
 F_2(\vec{x}) = \left\{\begin{array}{ll}
 x_k\ \ \ & \mbox{if $x_k \leq p$}\\
 x_j & \mbox{if $p \leq x_j$}\\
 p & \mbox{if $x_j < p < x_k$}
\end{array}\right.
\end{equation}
This implies the existence of a unique point $p \in (a, +\infty)$, such that
for all $(i | j, k)$-well-separated instances $\vec{x}$ with $x_i = a$, $F_2(\vec{x}) = \med(p, x_j, x_k)$.

For the first two cases of (\ref{eq:p-median}), we consider the instances $\vec{z}' = (\vec{z}_{-k}, p)$ and $\vec{z}'' = (\vec{z}_{-j}, p)$, which are both $(i | j, k)$-well-separated.
Since $F$ is strategyproof, $F_2(\vec{z}') = p = z'_k$. Then, by Proposition~\ref{prop:left_cover}, for every $(i | j, k)$-well-separated instance $\vec{x} = (\vec{z}_{-\{j,k\}}, x_j, x_k)$ with $x_k \leq p$, $F_2(\vec{x}) = x_k$.
Similarly, $F_2(\vec{z}'') = p = z''_j$, and by Proposition~\ref{prop:right_cover}, for every $(i | j, k)$-well-separated instance $\vec{x} = (\vec{z}_{-\{j,k\}}, x_j, x_k)$ with $x_j \geq p$, $F_2(\vec{x}) = x_j$.

For the third case of (\ref{eq:p-median}), we assume that there is an $(i | j, k)$-well-separated instance $\vec{x}$ with $x_i = a$, $x_j < p < x_k$, and $F_2(\vec{x}) = q \neq p$, and reach a contradiction (i.e., this case essentially establishes the uniqueness of $p$). Without loss of generality, we assume that $q > p$ (the case where $q < p$ is symmetric), and let $\eps \in (0, (q-p)/2)$ be appropriately small. To reach a contradiction, we exploit two $(i | j, k)$-well-separated instances, $\vec{y}$ and $\vec{y}'$, with $y_i = y'_i = a$, $y_j = x_j$, $y_k = q$, $y'_j = p$, and $y'_k = p+\eps$, such that $F_2(\vec{y}) = q = y_k$ and $F_2(\vec{y}') = p = y'_j$. Since $\vec{y}$ is an $(i | j, k)$-well-separated instance with $F_2(\vec{y}) = y_k$, Proposition~\ref{prop:left_cover} implies that for the $(i|j,k)$-well-separated instance $\vec{y}' = (\vec{y}_{-\{j,k\}}, p, p+\eps)$, with $y'_k < y_k$, $F_2(\vec{y}') = y'_k \neq p$, a contradiction.

To conclude the proof, it remains to show that indeed $F_2(\vec{y}) = q$ and $F_2(\vec{y}') = p$. To this end, we recall that $q \in [x_j, x_k]$, by Proposition~\ref{prop:interval}. For the former instance $\vec{y} = (\vec{x}_{-k}, q)$, we observe that it is $(i | j, k)$-well-separated because the distance of $x_j$ to $q$ is no greater than the distance of $x_j$ to $x_k$, and that $F_2(\vec{y}) = q$, because $F_2(\vec{x}) = q$ and $F$ is strategyproof. The latter instance $\vec{y}' = (\vec{z}_{-\{j, k\}}, p, p+\eps)$ is also $(i | j, k)$-well-separated, because $p > z_j$ and $\eps$ is chosen sufficiently small. Furthermore, $F_2(\vec{y}') = p$, because $F_2(\vec{z}) = p$, by hypothesis, $F_2(\vec{z}_{-j}, p) = p$, by $F$'s strategyproofness, and $F_2(\vec{z}_{-\{j,k\}}, p, p+\eps) = p$, by Proposition~\ref{prop:c_pulls_b}, because the instance $(\vec{z}_{-j}, p)$ is $(i|j,k)$-well-separated.

\smallskip\noindent {\bf Case II: $\forall \vec{z}\,F_2(\vec{z}) \in \{z_j, z_k\}$.}
Let $\vec{z}$ be any $(i|j,k)$-well-separated instance with $z_i = a$. We let $p = a$, if $F_2(\vec{z}) = z_j$\,%
\footnote{The uniqueness of $p$ for the case where $p = a$ follows from the fact that there are $(i|j,k)$-well-separated instances $\vec{z}$ with $z_i = a$ and $z_j$ arbitrarily close to $a$.},
and $p = \uparrow$, if $F_2(\vec{z}) = z_k$.
To conclude the proof, we show that if ($F_2(\vec{z}) = z_j$ and thus) $p = a$, for all $(i | j, k)$-well-separated instances $\vec{x}$ with $x_i = a$, $F_2(\vec{x}) = x_j$ (the case where $p = \uparrow$ and $F_2(\vec{x}) = x_k$ is symmetric).

To reach a contradiction, let us assume that there is a $(i|j,k)$-well-separated instance $\vec{x}$ with $x_i = a$ and $F_2(\vec{x}) \neq x_j$. Therefore, since we have assumed that $F_2(\vec{x}) \in \{x_j, x_k\}$, $F_2(\vec{x}) = x_k$. By Proposition~\ref{prop:b_pulls_c}, we can assume without loss of generality that $x_j$ can be arbitrarily close to $x_k$ (as long as $x_j$ stays on the left of $x_k$, so that $\vec{x}$ is $(i|j,k)$-well-separated). By Proposition~\ref{prop:right_cover}, for all $(i|j,k)$-well-separated instances $\vec{x}'$ with $x'_i = z_i = a$ and $z_j \leq x'_j$, $F_2(\vec{x}') = x'_j$. Therefore, we can assume that $x_j < z_j$, and since $x_j$ can be arbitrarily close to $x_k$, $x_j < x_k < z_j$. To conclude the proof, we show the following technical claim, which states that the existence of such an instance $\vec{x}$ implies that for all $(i|j,k)$-well-separated instances $\vec{y}$ with $y_i = a$ and $y_k > x_k$, $F_2(\vec{y}) = y_k$, a contradiction to the existence of $\vec{z}$.

\begin{myclaim}\label{cl:moving-backwards}
Let $F$ be any nice mechanism, let $i$ be any agent, and let $a$ be any location of $i$ such that for all $(i|j,k)$-well-separated instances $\vec{z}$ with $z_i = a$, $F_2(\vec{z}) \in \{z_j, z_k\}$. Then, if there exists an $(i|j,k)$-well-separated instance $\vec{x}$ with $x_i = a$ and $F_2(\vec{x}) = x_k$, for all $(i|j,k)$-well-separated instances $\vec{x}' = (\vec{x}_{-\{j,k\}}, x'_j, x'_k)$ with $x_k < x'_k$, $F_2(\vec{x}') = x'_k$.
\end{myclaim}

\begin{proof}[of Claim~\ref{cl:moving-backwards}]
We first show that for all $(i|j,k)$-well-separated instances $\vec{y} = (\vec{x}_{-\{j, k\}}, y_j, y_k)$ with $x_j \leq y_j < x_k < y_k$, $F_2(\vec{y}) = y_k$. For sake of contradiction, let us assume that there exists such an instance $\vec{y}$ for which $F_2(\vec{y}) \neq y_k$.
To establish a contradiction, we first observe that for the instance $\vec{x}' = (\vec{x}_{-j}, y_j)$, $F_2(\vec{x}') = x_k$, by Proposition~\ref{prop:b_pulls_c}, since $\vec{x}'$ is an $(i|j,k)$-well-separated instance.
Since $\vec{y} = (\vec{x}'_{-k}, y_k)$ and $F_2(\vec{y}) \neq y_k$, there exists a $y_k$-hole $(l, r)$ in the imageset $I_k(y_i, y_j)$ with $x_k \leq l < y_k$. Let $y'_k \in (l, y_k)$ be any point in the left half of the hole $(l, r)$, i.e. let $y'_k = \min\{ (y_k + 2l) / 3, (r+2l)/ 3 \}$. Due to $F$'s strategyproofness, $F_2(\vec{y}_{-k}, y'_k) = l$, because $l$ is the closest point to $y'_k$ in $I_k(y_i, y_j)$. This contradicts the hypothesis, since we have found an $(i|j,k)$-well-separated instance $\vec{y}' = (\vec{y}_{-k}, y'_k)$ with $F_2(\vec{y}') \not\in \{y'_j, y'_k\}$.

To conclude the proof, we inductively apply what we have shown above, i.e., that if for an $(i|j,k)$-well-separated instance $\vec{x}$, $F_2(\vec{x}) = x_k$, then for all $(i|j,k)$-well-separated instances $\vec{y} = (\vec{x}_{-\{j, k\}}, y_j, y_k)$ with $x_j \leq y_j < x_k < y_k$, $F_2(\vec{y}) = y_k$. The proof is similar to the proofs of Proposition~\ref{prop:right_cover} and Proposition~\ref{prop:left_cover}.

Let $x'_k > x_k$ be any point, let $d = x'_k - x_k$, let $\delta = (x_k - x_j)/2$, and let $\kappa = \ceil{d/\delta}$. For every $\l = 1, \ldots, \kappa, \kappa+1$, we inductively consider the instance $\vec{x}_\lambda = (\vec{x}_{-\{j, k\}}, x_j + (\lambda-1) \delta, x_k + (\lambda-1)\delta)$.
We observe that the instance $\vec{x}_\lambda$ is $(i | j, k)$-well-separated, because the distance of the locations of agents $j$ and $k$ is $2\delta$, while the distance of the locations of agents $i$ and $j$ is at least their distance in $\vec{x}$.
Therefore, for every $\l = 1, \ldots, \kappa, \kappa+1$, $F_2(\vec{x}_\lambda) = x_k + (\lambda-1)\delta$. Moreover, by Proposition~\ref{prop:left_cover}, for any $(i|j,k)$-well-separated instance $\vec{y} = (x_i, y_j, y_k)$ with $y_k \leq
x_k + \kappa\delta$, $F_2(\vec{y}) = y_k$. Therefore, since $x'_k \leq x_k + \kappa\delta$, for all $(i|j,k)$-well-separated instances $\vec{x}' = (\vec{x}_{-\{j,k\}}, x'_j, x'_k)$ with $x_k < x'_k$, $F_2(\vec{x}') = x'_k$.
\qed\end{proof}

With the proof of Claim~\ref{cl:moving-backwards}, we conclude the proof of the lemma.
\qed\end{proof}

At the conceptual level, the allocation of $F_2(\vec{x})$ to $i$-left-well-separated instances $\vec{x}$ corresponds to a GMVS applied to the subinstance $\vec{x}_{-i}$.
We recall that a GMVS applied to a set $N$ of agents is characterized by $2^{|N|}$ thresholds $\alpha_S$, one for each $S \subseteq N$ (see e.g. \cite[Definition~10.3]{SV07}), with $\alpha_\emptyset$ (resp. $\alpha_N$) equal to the smallest (resp. largest) value in the mechanism's range. Then, for every instance $\vec{x}$, the facility is allocated to $\max_{S \subseteq N}\min\{ \alpha_S, x_i \!:\! i \in S\}$.
Lemma~\ref{l:single-p} implies that there is a unique threshold $p_1$ (resp. $p_2$) such that for any $(i | j, k)$-well-separated (resp. $(i |k, j)$-well-separated) instance $\vec{x}$ with $x_i = a$, $F_2(\vec{x}) = \med(p_1, x_j, x_k)$ (resp. $F_2(\vec{x}) = \med(p_2, x_j, x_k)$).
Therefore, for any $i$-left-well-separated instance $\vec{x}$ with $x_i = a$, a GMVS applied to the subinstance $\vec{x}_{-i}$ is characterized by $\alpha_{\{j\}} = p_2$ and $\alpha_{\{k\}} = p_1$. Setting $\alpha_\emptyset = a$ and $\alpha_{\{j, k\}} =\,\uparrow$, it holds that
$F_2(\vec{x}) = \max\{ \min\{ x_j, p_2 \}, \min\{ x_k, p_1 \} \}$.

The following lemma %, whose detailed proof can be found in the Appendix, Section~\ref{app:s:i-separated},
shows that due to the bounded approximation ratio of $F$, either $p_1 =\,\uparrow$ or $p_2 =\,\uparrow$. If $p_2 =\,\uparrow$, $j$ has at least as much power as $k$, and becomes the preferred agent of the agent-location pair $(i, a)$ for $i$-left-well-separated instances. Then, the threshold $p$ of $(i, a)$ is equal to $p_1$, and $F_2(\vec{x}) = \max\{ x_j, \min\{ x_k, p \} \}$ (see also Fig.~\ref{fig:allocation}.a). We note that %the smaller the value of $p$, the greater the power of the preferred agent $j$. Thus,
the threshold $p$ essentially quantifies how much more powerful is the preferred agent $j$ than the third agent $k$. At the two extremes, if $p = a$, $j$ serves as a dictator on the right of $i$, while if $p =\,\uparrow$, the rightmost facility is allocated to the maximum of $x_j$ and $x_k$.
Similarly, if $p_1 =\,\uparrow$, $k$ is the preferred agent of the agent-location pair $(i, a)$ for $i$-left-well-separated instances, and $p = p_2$. Then, $F_2(\vec{x}) = \max\{ \min\{ x_j, p \}, x_k, \}$.

%Next, we show that for all $i$-left-well-separated instances, the rightmost facility is determined by a GMVS, whose characteristic ``thresholds'' may depend on the identity $i$ and the location $a$ of the isolated agent, but not on the relative order of $j$ and $k$.
%

\begin{lemma}\label{l:i-separated}
For any agent $i \in \{1, 2, 3\}$ and any location $a$, there are a unique threshold $p \in [a, +\infty) \union \{ \uparrow \}$ and a preferred agent $\ell \in \{1, 2, 3\} \setminus \{ i \}$ such that for any $i$-left-well-separated instance $\vec{x}$ with $x_i = a$, it holds that: %$F_2(\vec{x}) = x_\ell$, if $x_\ell \geq p$, and $F_2(\vec{x}) = \med(p, x_j, x_k)$, otherwise.
\[ F_2(\vec{x}) = \left\{
    \begin{array}{ll}
        x_\ell   & \mbox{if $x_\ell \geq p$}\\[1pt]
        \med(p, x_j, x_k)\ \ \ &\mbox{otherwise}
    \end{array}\right.\]
\end{lemma}

\begin{proof}
By Lemma~\ref{l:single-p}, for any agent $i \in \{1, 2, 3\}$ and any location $a \in \reals$\,:
\begin{itemize}
\item there is a unique $p_1 \in [a, +\infty) \union \{ \uparrow \}$ such that for any $(i | j, k)$-well-separated instance $\vec{x}$ with $x_i = a$, $F_2(\vec{x}) = \med(p_1, x_j, x_k)$, and

\item there is a unique $p_2 \in [a, +\infty) \union \{ \uparrow \}$ such that for any $(i | k, j)$-well-separated instance$\vec{x}$ with $x_i = a$, $F_2(\vec{x}) = \med(p_2, x_j, x_k)$.
\end{itemize}

We observe that it suffices to show that either $p_1 =\,\uparrow$ or $p_2 =\,\uparrow$. More specifically, if $p_1 =\,\uparrow$, then the preferred agent is $k$ and $p = p_2$. Then, if $x_k \geq p_2$, we distinguish between two cases depending on the order of $x_j$ and $x_k$. If $x_k \leq x_j$, then $\vec{x}$ is an $(i | k, j)$-well-separated instance, and $F_2(\vec{x})$ is located at the median of $p_2, x_k, x_j$, that is $x_k$. If $x_k > x_j$, $\vec{x}$ is an $(i | j, k)$-well-separated instance, and $F_2(\vec{x})$ is located at the median of $x_j, x_k, p_1$, that is $x_k$, because $p_1 =\,\uparrow$.
If $x_k < p_2$, we again distinguish between two cases. If $x_k \leq x_j$, then $\vec{x}$ is an $(i | k, j)$-well-separated instance, and $F_2(\vec{x}) = \med(p_2, x_k, x_j)$. If $x_k > x_j$, $\vec{x}$ is an $(i | j, k)$-well-separated instance, and $F_2(\vec{x})$ is located at the median of $x_j, x_k, p_1$, which coincides with $\med(x_j, x_k, p_2)$, because $p_1 =\,\uparrow$ and $p_2 > x_k > x_j$. If $p_2 =\,\uparrow$, then the preferred agent is $j$ and $p = p_1$, and the lemma follows from a similar analysis.

We proceed to establish that either $p_1 =\,\uparrow$ or $p_2 =\,\uparrow$. To reach a contradiction, we assume that both $p_1, p_2 \in [a, +\infty)$. Let $\vec{x}$ be an $(i| j, k)$-well-separated instance where both $x_j$ and $x_k$ exceed $\max\{ p_1, p_2\}$. Then, by Lemma~\ref{l:single-p}, $F_2(\vec{x}) = \med(p_1, x_j, x_k) = x_j$. Since $F_2(\vec{x}) \neq x_k$, there is a $x_k$-hole $(l, r)$ in the image set $I_k(x_i, x_j)$. For some appropriately small $\eps > 0$, we consider the instances $\vec{x}' = (\vec{x}_{-k}, r-\eps)$ and $\vec{x}'' = ({\vec{x}'}_{-j}, r)$. Since $r \in I_k(x_i, x_j)$, $F_2(\vec{x}') = r$. Since $r \in I_j(x_i, r-\eps)$, $F_2(\vec{x}'') = r$. On the other hand, since $\eps$ is chosen appropriately small, the instance $\vec{x}''$ is $(i|k, j)$-well-separated. Therefore, by Lemma~\ref{l:single-p}, $F_2(\vec{x}'') = \med(p_2, r-\eps, r) \neq r$, a contradiction.
\qed\end{proof}

By a symmetric argument, we can establish the symmetric version of Lemma~\ref{l:i-separated} for $i$-right-well-separated instances. The only difference is that $p \in \{ \downarrow \} \union (-\infty, a]$, and that $F_1(\vec{x}) = x_\ell$, if $x_\ell \leq p$, and $F_1(\vec{x}) = \med(p, x_j, x_k)$, otherwise. We highlight that an agent location pair $(i, a)$ may have a different preferred agent and a different threshold for $i$-left and $i$-right well-separated instances.

Lemma~\ref{l:i-separated} implies that every nice mechanism $F$ admits a function $\l_F$ (resp. $\r_F$) % : \{ 1, 2, 3\} \times \reals \mapsto \{ 1, 2, 3 \} \times (\reals \union \{ \uparrow \})$
that maps any agent-location pair $(i, a)$ to a pair $(j, p)$, where $p \in [a, +\infty) \union \{ \uparrow \}$ (resp. $p \in \{ \downarrow \} \union (-\infty, a]$) is the unique threshold and $j \in \{1, 2, 3\} \setminus \{ i \}$ is the preferred agent of $(i, a)$ satisfying the condition of (resp. the symmetric version of) Lemma~\ref{l:i-separated} for $i$-left (resp. $i$-right) well-separated instances.
%
%Similarly, every nice mechanism $F$ admits a function $\r_F : \{ 1, 2, 3\} \times \reals \mapsto \{ 1, 2, 3 \} \times (\reals \union \{ \downarrow \})$ that maps any agent-location pair $(i, a)$ to a pair $(\ell, p)$, where $p \in \{ \downarrow \} \union (-\infty, a]$ is the unique threshold and $\ell \in \{1, 2, 3\} \setminus \{ i \}$ is the preferred agent of $(i, a)$ satisfying the condition of the symmetric version of Lemma~\ref{l:i-separated} for $i$-right-well-separated instances.
%
We note that the preferred agent is uniquely determined in the proof of Lemma~\ref{l:i-separated}, unless $p =\,\uparrow$ for $i$-left (resp. $p =\,\downarrow$ for $i$-right) well-separated instances. If $p = \,\uparrow$ (resp. $p =\,\downarrow$), any agent different from $i$ can play the role of the preferred agent, since this choice does not affect the allocation of $F_2(\vec{x})$ (resp. $F_1(\vec{x})$). For convenience, we simply write $\l(i, a)$ and $\r(i, a)$, %instead of $\l_F(i, a)$ and $\r_F(i, a)$,
%whenever the mechanism $F$ is clear from the context. Moreover, we
and let $\l_p(i, a)$ and $\l_\ell(i, a)$ (resp. $\r_p(i, a)$ and $\r_\ell(i, a)$) denote the threshold $p$ and the preferred agent $j$ of $(i, a)$ for $i$-left (resp. $i$-right) well-separated instances.
Lemma~\ref{l:i-separated} implies that:

%Before we proceed to the second step of the proof, we state the following corollary of Lemma~\ref{l:i-separated}.

\begin{corollary}\label{cor:dict-location}
Let $\vec{x}$ be any $(i|j,k)$-well-separated instance such that $F_2(\vec{x}) = x_j$. Then $\l_\ell(i, x_i) = j$ and $\l_p(i, x_i) \leq x_j$.
\end{corollary}

\begin{proof}
We show that for any $(i|j,k)$-well-separated instance $\vec{x}$, the rightmost facility is allocated to the middle agent $j$ only if $j$ is the preferred agent of $(i, x_i)$ and is located on the left of $(i, x_i)$'s threshold. This can be easily verified by Fig.~\ref{fig:allocation}.a, where the rightmost facility is never allocated to $k$ below the $x_j = x_k$ line, where $k$ is the middle agent, and is allocated to the preferred agent $j$ above the $x_j = x_k$ line, where $j$ is the middle agent, only on the right of $p$.

We also give a more formal argument. For convenience, we let $p_i \equiv \l_p(i, x_i)$. We first observe that if $\l_\ell(i, x_i) = k$, $F_2(\vec{x}) = x_k$. This follows from Lemma~\ref{l:i-separated}, because $x_j < x_k$. More specifically, either $x_k \geq p_i$, and the rightmost facility is allocated to $k$, as the preferred agent, or $x_k < p_i$, and the rightmost facility is allocated to $x_k$, as the median of $p_i$, $x_j$, and $x_k$. Hence, $\l_\ell(i, x_i) = j$.
Moreover, if $p_i > x_j$, $F_2(\vec{x}) = \med(p_i, x_j, x_k) > x_j$, by the second case of Lemma~\ref{l:i-separated}. Therefore, $p_i \leq x_j$.
\qed\end{proof}

%A detailed proof of Corollary~\ref{cor:dict-location} can be found in the Appendix, Section~\ref{app:s:dict-location}. By a symmetric argument, we obtain the symmetric version of the corollary: for any $(i,j|k)$-well-separated instance $\vec{x}$, the leftmost facility cannot be allocated to the middle agent $j$, unless $\r_\ell(i, x_i) = j$ and $\r_p(i, x_i) \geq x_j$.

\subsection{General Instances and the Range of the Threshold $p$}
\label{s:gen-instances}

Next, we proceed to show that $\l_p(i, a) \in \{a, \uparrow\}$ (and by symmetry, $\r_p(i, a) \in \{a, \downarrow\}$). As before, we consider the instances from left to right, and state and prove our lemmas in terms of $\l(i, a)$. It is straightforward to verify that the symmetric statements hold for $\r(i, a)$.

We start with three useful technical lemmas that essentially extend the allocation of Lemma~\ref{l:i-separated} to instances that are not necessarily well-separated (see also Fig.~\ref{fig:allocation}.b).
%
%Using these lemmas, we determine where the location of the facility serving (at least one of) the agents $j$ and $k$ for most of the undecided areas in Fig.~\ref{fig:allocation}.a, and obtain Fig.~\ref{fig:allocation}.b.
%
The first lemma shows that the preferred agent essentially serves as a dictator when located on the right of the threshold $p$.
%
%The proofs of Lemma~\ref{l:non-separated-inside} and Lemma~\ref{l:approx} are similar to the proof of Lemma~\ref{l:non-separated} and can be found in the Appendix, Section~\ref{app:s:non-separated-inside}, and Section~\ref{app:s:approx}, respectively.

%
%Their proofs are similar and can be found in the Appendix, Sections~\ref{app:s:non-separated}, \ref{app:s:non-separated-inside}, and \ref{app:s:approx}, respectively. In all of them, the proof idea is to assume an instance that violates the lemma, and taking advantage of the corresponding hole in the image set, to construct a well-separated instance for which $F$ has an unbounded approximation ratio.

\begin{lemma}\label{l:non-separated}
Let $i$ be any agent and $a$ be any location, and let $\l(i, a) = (j, p)$. For any instance $\vec{x}$ with $a = x_i < \min\{x_j, x_k\}$, and $x_j \geq p$, it holds that $x_j \in F(\vec{x})$.
\end{lemma}

\begin{proof}
%Let us fix an agent $i$ and some location $a$. By Lemma~\ref{l:i-separated}, there are a threshold $p \in [a, +\infty) \union \{ \uparrow \}$ and a preferred agent $j \in \{1, 2, 3\} \setminus \{ i \}$ such that for any $i$-left-well-separated instance $\vec{y}$ with $y_i = a$ and $y_j \geq p$, $y_j \in F_2(\vec{y})$. We next show that for the particular $p$ and $j$, and for any instance $\vec{x}$ with $x_i = a$ and $x_j \geq p$, $x_j \in F(\vec{x})$.
%
We fix an agent $i$ and some location $a$, and let $\l(i, a) = (j, p)$. For sake of contradiction, we assume that there is an instance $\vec{x}$, with $x_i = a$ and $x_j \geq p$, such that $x_j \not\in F(\vec{x})$. Therefore, $x_j$ does not belong to the image set $I_j(x_i, x_k)$, and there is a $x_j$-hole $(l, r)$ in $I_j(x_i, x_k)$. For a small $\eps \in (0, \min\{(r-l)/2, r-x_j\})$, we consider the instance $\vec{x}' = (\vec{x}_{-j}, r-\eps)$. Since $r-\eps$ is in the right half of the $x_j$-hole $(l, r)$, $r \in F(\vec{x}')$. Now, we consider the instance $\vec{x}'' = (\vec{x}_{-\{j,k\}}, r-\eps, r)$, and show that $F(\vec{x}'') = (r-\eps, r)$. On the one hand, $r \in F(\vec{x}'')$, because $F$ is strategyproof. On the other hand, we can choose $\eps$ small enough that the instance $\vec{x}''$ is $(i|j,k)$-well-separated. Then, by Lemma~\ref{l:i-separated}, $x''_j \in F(\vec{x}'')$, because by the choice of $\eps$, $x''_j = r-\eps > x_j \geq p$. However, this implies that $x_i$ is served by the facility at $r-\eps$, which contradicts $F$'s bounded approximation ratio.
\qed\end{proof}

By a symmetric argument, we obtain the symmetric version of Lemma~\ref{l:non-separated} for where $i$ is the rightmost agent: if $\r(i, a) = (j, p)$, for all instances $\vec{x}$ with $x_i = a$, $x_i > \max\{ x_j, x_k \}$, and $x_j \leq p$, $x_j \in F(\vec{x})$.

The next lemma shows that if both agents $j$ and $k$ are located (on the right of agent $i$ and) on the left of $p$, the rightmost facility is allocated to the rightmost agent. %The proof is similar to the proof of Lemma~\ref{l:non-separated}, and can be found in the Appendix, Section~\ref{app:s:non-separated-inside}.

\begin{lemma}\label{l:non-separated-inside}
Let $i$ be any agent and $a$ be any location, and let $\l_p(i, a) = p$. For any instance $\vec{x}$ with $x_i = a$ and $x_j, x_k \in (x_i, p]$, it holds that $F_2(\vec{x}) = \max\{x_j, x_k\}$.
\end{lemma}

\begin{proof}
%Let us fix an agent $i$ and a location $a \in \reals$. By Lemma~\ref{l:i-separated}, there are a threshold $p \in [a, +\infty) \union \{ \uparrow \}$ such that for any $i$-left-well-separated instance $\vec{y}$ with $y_i = a < y_j, y_k \leq p$, $F_2(\vec{y}) = \med(y_j, y_k, p) = \max\{y_j, y_k\}$. We next show that this implies that for any instance $\vec{x}$ with $x_i = a < x_j, x_k \leq p$, $F_2(\vec{x}) = \max\{x_j, x_k\}$.
%
We fix an agent $i$ and a location $a \in \reals$, and let $\l_p(i, a) = p$. For sake of contradiction, let us assume that there is an instance $\vec{x}$, with $a= x_i < x_j < x_k \leq p$, such that $F_2(\vec{x}) \neq x_k$ (the case where $a < x_k < x_j \leq p$ is symmetric, while if $a < x_j = x_k \leq p$, then $F_2(\vec{x}) = x_k$, because $F$ has a bounded approximation ratio).
Therefore, $x_k$ does not belong to the image set $I_k(x_i, x_j)$, and there is a $x_k$-hole $(l, r)$ in $I_k(x_i, x_j)$.
%
%We note that $x_j \leq l$, since $x_j \in F(x_i, x_j, x_j)$, due to $F$'s bounded approximation ratio.
%
For an appropriately small $\eps \in (0, \min\{(r-l)/2, x_k-l\})$, we consider the instance $\vec{x}' = (\vec{x}_{-k}, l+\eps)$. Since $l+\eps$ is in the left half of the $x_k$-hole, $l \in F(\vec{x}')$.
Now, we consider the instance $\vec{x}'' = (\vec{x}_{-\{j,k\}}, l, l+\eps)$, and show that $F$ places its two facilities at $l$ and $l+\eps$. On the one hand, $l \in F(\vec{x}'')$, because $F$ is strategyproof. On the other hand, we can choose $\eps$ small enough that the instance $\vec{x}''$ is $(i|j,k)$-well-separated. Then, by Lemma~\ref{l:i-separated}, $x''_k \in F(\vec{x}'')$, because by the choice of $\eps$, $x''_j < x''_k < p$. However, this implies that $x_i$ is served by the facility at $l$, which contradicts $F$'s bounded approximation ratio.
\qed\end{proof}

By a symmetric argument, we obtain the symmetric version of Lemma~\ref{l:non-separated-inside} for instances where $i$ is the rightmost agent: if $\r_p(i, a) = p$, for all instances $\vec{x}$ with $x_i = a > x_j, x_k \geq p$, $F_1(\vec{x}) = \min\{x_j, x_k\}$.

%The proof of Lemma~\ref{l:non-separated-inside} is similar to the proof of Lemma~\ref{l:non-separated}, and can be found in the Appendix, Section~\ref{app:s:non-separated-inside}.
%
The next lemma complements the previous two, and shows that the preferred agent is allocated a facility even if she lies on the left of $p$ and is not the rightmost agent, provided that the distance of the rightmost agent to $p$ is large enough.

\begin{lemma}\label{l:approx}
Let $\rho$ be the approximation ratio of $F$, let $i$ be any agent and $a$ be any location, and let $\l(i, a) = (j, p)$. Then, for any instance $\vec{x}$ with $a = x_i < x_j < x_k$, and $x_k - p > \rho(p - x_i)$, %it holds that
$x_j \in F(\vec{x})$.
\end{lemma}

\begin{proof}
We fix an agent $i$ and some location $a$, and let $\l(i, a) = (j, p)$.
%
%By Lemma~\ref{l:non-separated}, there are a $p \in [a, +\infty) \union \{ \uparrow \}$ and an agent $j \in \{1, 2, 3\} \setminus \{ i \}$ such that for any instance $\vec{y}$ with $y_i = a$ and $y_j \geq p$, $y_j \in F_2(\vec{y})$. We next show that for the particular $p$ and $j$, and for any instance $\vec{x}$ with $x_i = a$, $x_i < x_j < x_k$ and $x_k - p > \rho(p - x_i)$, $x_j \in F(\vec{x})$.
%
For sake of contradiction, let us assume that there is an instance $\vec{x}$ with $a = x_i < x_j < x_k$ and $x_k - p > \rho(p - x_i)$, for which $x_j \not\in F(\vec{x})$. Therefore, $x_j$ does not belong to the image set $I_j(x_i, x_k)$, and there is a $x_j$-hole $(l, r)$ in $I_j(x_i, x_k)$. Moreover, $x_j < p$ and $r \leq p$, because by Lemma~\ref{l:non-separated}, for any $y \geq p$, $y \in F(\vec{x}_{-j}, y)$. Since $F$ is strategyproof, for any point $x'_j$ in the right half of the hole $(l, r)$, $r \in F(\vec{x}_{-j}, x'_j)$. By Proposition~\ref{prop:middle}, $F_1(\vec{x}_{-j}, x'_j) \leq x'_j$, and thus $r = F_2(\vec{x}_{-j}, x'_j)$. Therefore, $x_k$ is served by the facility at $r$, and $\cost[F(\vec{x}_{-j}, x'_j)] \geq x_k - p$. This contradicts the hypothesis that the approximation ratio is $\rho$, since the optimal cost for $(\vec{x}_{-j}, x'_j)$ is at most $x'_j - x_i \leq p - x_i < (x_k - p) / \rho$.
\qed\end{proof}

Now, we are ready to show that the threshold $p$ always falls in the two extremes.%, thus making a major step towards the proof of Theorem~\ref{thm:3agents}.

\begin{lemma}\label{l:i-cases}
For any agent $i$ and location $a$, $\l_p(i, a) \in \{a, \uparrow \}$ and $\r_p(i, a) \in \{a, \downarrow \}$.
\end{lemma}

\begin{proof}
We show that $\l_p(i, a) \in \{a, \uparrow \}$. The claim that $\r_p(i, a) \in \{a, \downarrow \}$ follows by a symmetric argument.
For sake of contradiction, let us assume that there exist an agent $i$ and a location $a$, for which $p_i \equiv \l_p(i, a) \in (a, +\infty)$, and let $j = \l_\ell(i, a)$.
To reach a contradiction, we study the preferred agents of four appropriately chosen well-separated instances $\vec{x}$, $\vec{y}$, $\vec{z}$, and $\vec{w}$. Intuitively, exploiting the properties of the instances $\vec{x}$, $\vec{y}$, $\vec{z}$, we show that in the last instance $\vec{w}$, where the agents are arranged according to the permutation $(i, j, k)$, the preferred agent of $(i, w_i)$ on the right is $j$ and the preferred agent of $(k, w_k)$ on the left is $i$. Moreover, $\vec{w}$ is chosen to satisfy the conditions of Lemma~\ref{l:approx} for $i$ and her preferred agent $j$ (see also Fig.~\ref{fig:instances}). This implies a facility allocation for $\vec{w}$ with an approximation ratio larger than $F$'s approximation ratio $\rho$, a contradiction.

%%%%%%%%%%%%%%%%%%%%%%%%%%%%%%%%%%%%%%%%%%%%%%%%%%%%%%%%%%%%%%%%%%%%%
\insertfig{0.75\textwidth}{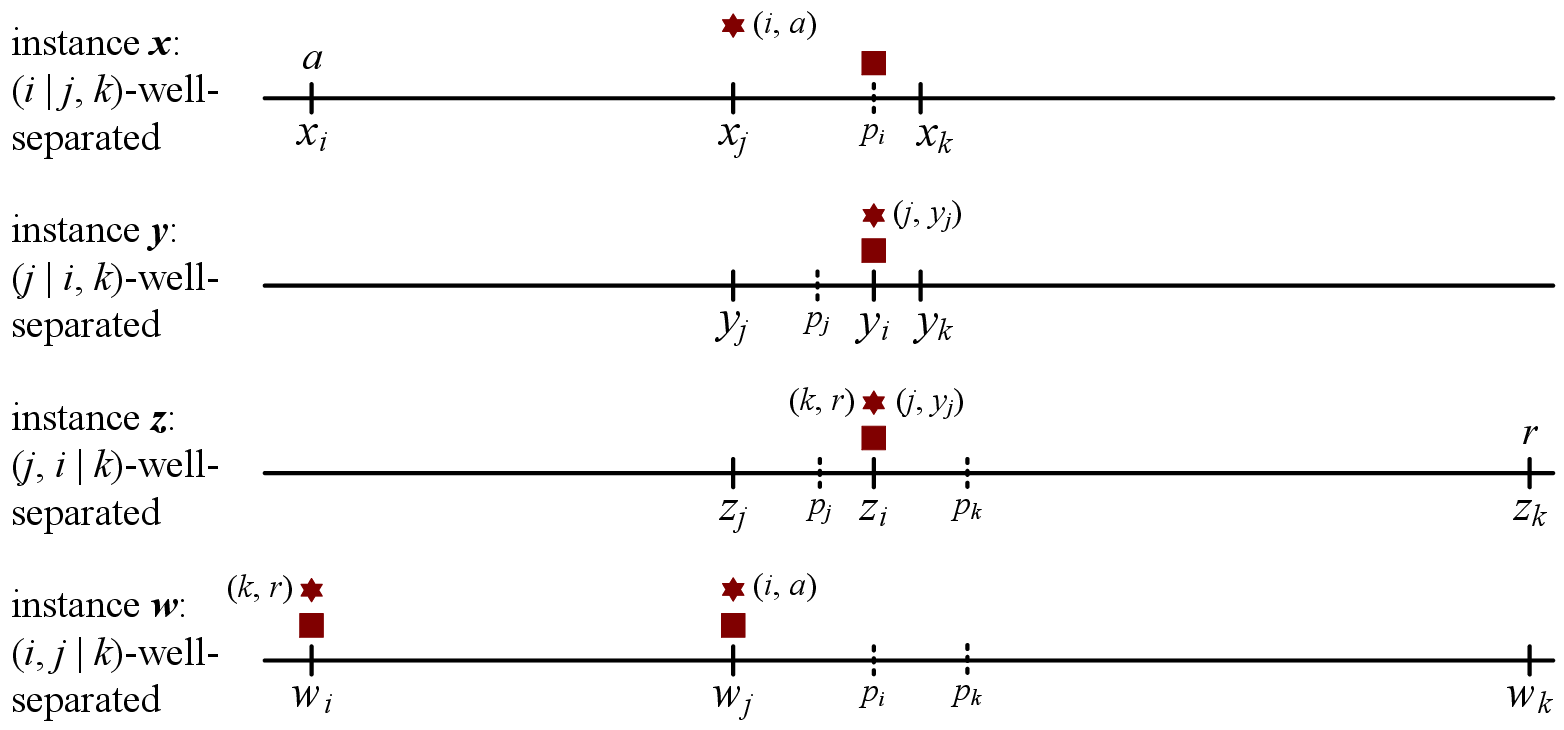}{\label{fig:instances}%
The four instances considered in the proof of Lemma~\ref{l:i-cases} with the corresponding thresholds and preferred agents. A square indicates that a facility is placed at the corresponding location. A star indicates that the corresponding agent is the preferred agent of the agent-location pair appearing next to the star.}
%%%%%%%%%%%%%%%%%%%%%%%%%%%%%%%%%%%%%%%%%%%%%%%%%%%%%%%%%%%%%%%%%%%%%

Formally, let the first instance $\vec{x}$ be any $(i|j,k)$-well-separated instance such that $x_j < p_i < x_k$ and the instance $(\vec{x}_{-i}, p_i)$ is $(j|i,k)$-well-separated. Then, by Lemma~\ref{l:i-separated}, $F_2(\vec{x}) = \med(p_i, x_j, x_k) = p_i$.
The second instance is $\vec{y} = (\vec{x}_{-i}, p_i)$, which is $(j|i,k)$-well-separated, by the choice of $\vec{x}$. Moreover, $F_2(\vec{y}) = p_i = y_i$, due to $F$'s strategyproofness. We let $p_j \equiv \l_p(j, y_j)$. Then, by Corollary~\ref{cor:dict-location}, $\l_\ell(j, y_j) = i$ and $p_j \leq y_i$.

The third instance is $\vec{z} = (\vec{y}_{-k}, r)$, where $r$ is chosen so that $r - p_i > \rho(p_i - a)$, where $\rho$ is the approximation ratio of $F$. Therefore, $\vec{z}$ is a $(j, i |k)$-well-separated instance. Since $y_j = z_j$, $\l(j, z_j) = (i, p_j)$ and $p_j \leq z_i$. Hence, we apply Lemma~\ref{l:non-separated} to the instance $\vec{z}$, and obtain that $z_i \in F(\vec{z})$. Therefore, by (the symmetric of) Corollary~\ref{cor:dict-location} applied to the $(j, i |k)$-well-separated instance $\vec{z}$, $\r_\ell(k, r) = i$ and $\r_p(k, r) \equiv p_k \geq z_i$.

The fourth instance is $\vec{w} = (\vec{x}_{-k}, r)$. Since $w_i = a$, we have that $\l(i, w_i) = (j, p_i)$. Since the location $r$ of $k$ in $\vec{w}$ is chosen so that $r - p_i > \rho(p_i - w_i)$, we apply Lemma~\ref{l:approx} to $\vec{w}$, and obtain that $w_j \in F(\vec{w})$. Moreover, by the choice of $r$, $\vec{w}$ is an $(i, j|k)$-well-separated instance. Therefore, since $\r_\ell(k, r) = i$ and $\r_p(k, r) \geq z_i \geq w_i$, by (the symmetric of) Lemma~\ref{l:i-separated} applied to the $k$-right-well-separated instance $\vec{w}$, $F_1(\vec{w}) = w_i$. By the choice of $r$, the hypothesis that $F$'s approximation ratio is $\rho$, since $\cost[F(\vec{w})] = r - w_j > r - p_i$, while the optimal cost for $\vec{w}$ is $w_j - w_i < p_i - a < (r - p_i) / \rho$.
\qed\end{proof}

\subsection{On the Existence of a Partial Dictator}
\label{s:3agent-fin}

%In this section, we conclude the proof of Theorem~\ref{thm:3agents}.
%
Having restricted the range of $p$ to the two extremes, we can make some useful observations about the preferred agents and the thresholds of the leftmost and the rightmost locations. We start with the following consequence of Lemma~\ref{l:i-cases} and Corollary~\ref{cor:dict-location}.

\begin{proposition}\label{prop:dict-location-a}
If there is an $(i|j,k)$-well-separated instance $\vec{x}$, such that $F_2(\vec{x}) = x_j$, then $\l(i, x_i) = (j, x_i)$. If there is an $(i,j|k)$-well-separated instance $\vec{x}$, such that $F_1(\vec{x}) = x_j$, then $\r(k, x_k) = (j, x_k)$.
\end{proposition}

The following lemma shows that the preferred agent and the threshold of an agent $i$ can only depend on her identity, and not on her location.

\begin{lemma}\label{l:always-a}
Let $i$ be an agent and $a$ be a location such that $\l(i, a) = (j, a)$ (resp. $\r(i, a) = (j, a)$). Then, for all locations $b \in \reals$, $\l(i, b) = (j, b)$ (resp. $\r(i, b) = (j, b)$).
\end{lemma}

\begin{proof}
We only show that if $\l(i, a) = (j, a)$, then $\l(i, b) = (j, b)$, for all locations $b$. The other case follows from a symmetric argument.
Let $i$ be an agent and $a$ be a location such that $\l(i, a) = (j, a)$, for some agent $j$, and let $b$ be any location. To show that $\l(i, b) = (j, b)$, we consider two appropriate instances $\vec{x}$ and $\vec{y}$.

The instance $\vec{x}$ is defined as $x_i = a$, $x_j = a + 1$, and $x_k = r > \max\{a+1, b\}$, where $r$ is chosen large enough that $\vec{x}$ is an $(i,j|k)$-well-separated instance. Since $x_j \geq \l_p(i, a)$,  Lemma~\ref{l:non-separated} implies that $x_j \in F(\vec{x})$. Then, since $\vec{x}$ is an $(i,j|k)$-well-separated instance, by Proposition~\ref{prop:dict-location-a}, $\r(k, r) = (j, r)$.

The instance $\vec{y}$ is defined as $y_i = b$, $y_j = r-\eps$, and $y_k = r$, where $\eps > 0$ is chosen so small that $\vec{y}$ is $(i|j,k)$-well-separated. Since $y_j \leq \r_p(k, r)$, the symmetric version of Lemma~\ref{l:non-separated} applied to $\vec{y}$ implies that $y_j \in F(\vec{y})$. Then, since the instance $\vec{y}$ is $(i|j,k)$-well-separated, by Proposition~\ref{prop:dict-location-a},
$\l(i, b) = (j, b)$.
 \qed\end{proof}

Next, Lemma~\ref{l:others-a} shows that if an agent $i$ imposes another agent $j$ as a partial dictator, the third agent $k$ agrees with $i$ not only on the existence of a partial dictator, but also on the partial dictator's identity.

\begin{lemma}\label{l:others-a}
If there exists an agent $i$ and a location $a$ such that $\l(i, a) = (j, a)$ (resp. $\r(i, a) = (j, a)$), then for the third agent $k$ and all locations $b \in \reals$, it holds that $\r(k, b) = (j, b)$ (resp. $\l(k, b) = (j, b)$).
\end{lemma}

\begin{proof}
We only consider the case where $\l(i, a) = (j, a)$, and show that $\r(k, b) = (j, b)$, for the third agent $k$ and all locations $b$. The symmetric case follows from a symmetric argument.
Let $i$ be an agent and $a$ be a location such that $\l(i, a) = (j, a)$, for some agent $j$, and let $b$ be any location of the third agent $k$. Without loss of generality, we assume that $a < b$. Otherwise, we select a location $a'$ of $i$ with $a' < b$, and have that $\l(i, a') = (j, a')$, by Lemma~\ref{l:always-a}. To establish the lemma, we consider an instance $\vec{x}$ defined as $x_i = a$, $x_k = b$, and $x_j = a+\eps$, where $\eps > 0$ is chosen small enough that $\vec{x}$ is an $(i,j|k)$-well-separated instance. Since $x_j \geq \l_p(i, a)$, Lemma~\ref{l:non-separated} implies that $x_j \in F(\vec{x})$. Then, since $\vec{x}$ is an $(i,j|k)$-well-separated instance, by Proposition~\ref{prop:dict-location-a}, $\r(k, b) = (j, b)$.
\qed\end{proof}

Lemma~\ref{l:only-dictator} below shows that if there is an agent $j$ imposed as a partial dictator by the leftmost and the rightmost agent, then $j$ is the only agent with this property. Thus, if $j$ is located at the one extreme of the instance, the other facility is placed at the other extreme.

\begin{lemma}\label{l:only-dictator}
If there exists an agent $i$ and a location $a$ such that $\l(i, a) = (j, a)$, then for all locations $b \in \reals$, it holds that $\l_p(j, b) =\,\uparrow$ and $\r(j, b) =\,\downarrow$.
\end{lemma}

\begin{proof}
We only show that for all locations $b$, $\l(j, b) =\,\uparrow$. The claim that $\r(j, b) =\,\downarrow$ follows from a symmetric argument.
In the following, we let $(i, a)$ be any agent-location pair with $\l_p(i, a) = (j, a)$, and let $k$ be the third agent. We observe that for all locations $a' \in \reals$, $\l(i, a') = (j, a')$, by Lemma~\ref{l:always-a}, and $\r(k, a') = (j, a')$, by Lemma~\ref{l:others-a}.
For sake of a contradiction, we assume that for some location $b \in \reals$, $\l_p(j, b) = b$. Therefore, there exists a (unique) preferred agent $\ell \in \{i, k\}$ such that $\l(j, b) = (\ell, b)$, and by Lemma~\ref{l:always-a}, for all locations $b' \in \reals$, $\l(j, b') = (\ell, b')$. Moreover, for all instances $\vec{x}$ with $x_j = b < \min\{ x_i, x_k \}$, Lemma~\ref{l:non-separated} implies that $x_\ell \in F(\vec{x})$. To reach a contradiction, we next show that $\l_\ell(j, b) \neq i$ and that $\l_\ell(j, b) \neq k$.

To this end, we first assume that $\l_\ell(j, b) = i$, and consider an instance $\vec{y}$ with $y_j = 0$, $y_k = 1$, and $y_i = \eps$, where $\eps > 0$ is chosen small enough that $\vec{y}$ is an $(j,i|k)$-well-separated instance.
Since $\l(j, 0) = (i, 0)$, since $j$ is the leftmost agent, and since $y_i > 0$, Lemma~\ref{l:non-separated} implies that $y_i \in F(\vec{x})$.
Moreover, since $\r(k, 1) = (j, 1)$, $k$ is the rightmost agent, and $y_j < 1$, the symmetric version of Lemma~\ref{l:non-separated} implies that $y_j \in F(\vec{x})$. Therefore, the agent $k$ is served by the facility at $y_i$, which contradicts $F$'s bounded approximation ratio, because the instance $\vec{y}$ is $(j,i|k)$-well-separated. %Consequently, $\l_\ell(j, b) \neq i$.

We have also to consider the case where $\l_\ell(j, b) = k$. In this case, we consider an instance $\vec{z}$ with $z_i = 0$, $z_j = 1$, and $z_k = 1+\eps$, where $\eps > 0$ is chosen small enough that $\vec{z}$ is an $(i|j,k)$-well-separated instance.
Since $\l(i, 0) = (j, 0)$, $i$ is the leftmost agent in $\vec{z}$, $z_j > 0$, and the instance $\vec{z}$ is $(i|j,k)$-well-separated, Lemma~\ref{l:i-separated} implies that $F_2(\vec{z}) = z_j$.
Therefore, $z_k \not\in F(\vec{z})$, and there is a $z_k$-hole $(l, r)$ in the imageset $I_k(z_i, z_j)$ with $l = 1$. We now consider the instance $\vec{z}' = (\vec{z}_{-\{i,k\}}, r, r-\delta)$, where $\delta > 0$ is chosen small enough that $r-\delta$ is in the right half of the hole $(l, r)$ and $\vec{z}'$ is an $(j|k,i)$-well-separated instance. We observe that $F(\vec{z}') = (r-\delta, r)$.
More specifically, $r = z'_i \in F(\vec{z}')$, due to $F$'s strategyproofness. Otherwise, the agent $i$ could switch to location $z_i$ and have $r \in F(\vec{z}'_{-i}, z_i)$, since $r$ is the closest point to $z'_k = r-\delta$ in the imageset $I_k(z_i, z_j)$.
Furthermore, since $\l(j, 1) = (k, 1)$, since $j$ is the leftmost agent in $\vec{z}'$, and since $z'_k > 1$, Lemma~\ref{l:non-separated} implies that $r-\delta = z'_k \in F(\vec{z}')$.
Therefore, the agent $j$ is served by the facility at $z'_k$, which contradicts $F$'s bounded approximation ratio, because the instance $\vec{z}'$ is $(j|k,i)$-well-separated. %Hence, $\l_\ell(j, b) \neq k$.
\qed\end{proof}

Now, we are are ready to conclude the proof of Theorem~\ref{thm:3agents}. We show that the behavior of any nice mechanism $F$ is essentially characterized by whether there are two agents that agree on imposing the third agent as a partial dictator or not. Theorem~\ref{thm:3agents} is an immediate consequence of the following lemma.

\begin{lemma}\label{l:3agents}
Either for all instances $\vec{y}$, $F(\vec{y}) \in ( \min \vec{y}, \max \vec{y} )$, or there exists an instance $\vec{x}$, with $x_i < x_j < x_k$, such that $x_j \in F(\vec{x})$. In the latter case:
\begin{itemize}
\item For all instances $\vec{y}$ with $y_i < y_j < y_k$, $y_j \in F(\vec{y})$.

\item Either for all instances $\vec{y}$ with $y_k < y_j < y_i$, $y_j \in F(\vec{y})$, or for all instances $\vec{y}$ with $y_k < y_j < y_i$,
    $F(\vec{y}) = ( \min \vec{y}, \max \vec{y} )$

\item For all the remaining instances $\vec{y}$, $F(\vec{y}) = ( \min \vec{y}, \max \vec{y} )$
\end{itemize}
\end{lemma}

\begin{proof}
We distinguish between the case where for all agents $i \in \{ 1, 2, 3\}$ and for all locations $a \in \reals$, $\l_p(i, a) =\,\uparrow$, and the case where there exists an agent $i \in \{ 1, 2, 3\}$ and a location $a \in \reals$ such that $\l_p(i, a) = a$. Lemma~\ref{l:i-cases} implies that these two cases are indeed complementary to each other.

\smallskip\noindent{\bf Case I: $\forall i \forall a \, \l_p(i, a) =\,\uparrow$.}
We show that in this case, $F$ always places the facilities at two extremes.
To this end, we let $\vec{x}$ be any instance, and let $i$ be the leftmost agent and $k$ be the rightmost agent in $\vec{x}$. By hypothesis, $\l_p(i, x_i) =\,\uparrow$. Therefore, by Lemma~\ref{l:non-separated-inside}, $F_2(\vec{x}) = x_k$.
Moreover, $\r_p(k, x_k) =\,\downarrow$, since otherwise, it would be $\r_p(k, x_k) = x_k$, by Lemma~\ref{l:i-cases}, and thus $\l_p(i, x_i) = x_i$, by Lemma~\ref{l:others-a}, which contradicts the hypothesis. Therefore, by the symmetric version of Lemma~\ref{l:non-separated-inside}, $F_1(\vec{x}) = x_i$.

\smallskip\noindent{\bf Case II: $\exists i \exists a \, \l_p(i, a) = a$.}
We let $(i, a)$ be any agent-location pair with $\l_p(i, a) = a$, let $j = \l_\ell(i, a)$ be the preferred agent of $i$, and let $k$ be the third agent. Therefore, for all locations $b$, $\l_p(j, b) =\,\uparrow$ and $\r_p(j, b) =\,\downarrow$, by Lemma~\ref{l:only-dictator}. We next show that in this case, the agent $j$ serves as a partial dictator, and satisfies the latter case of the conclusion.

We first observe that for all locations $a' \in \reals$, $\l(i, a') = (j, a')$, by Lemma~\ref{l:always-a}, and $\r(k, a') = (j, a')$, by Lemma~\ref{l:others-a}. Therefore, for all instances $\vec{x}$ with $x_i < x_k$, $x_j \in F(\vec{x})$. More specifically, if $x_i < x_j$, this follows from Lemma~\ref{l:non-separated}, while if $x_j < x_k$, this follows from the symmetric version of Lemma~\ref{l:non-separated}. Then, using Lemma~\ref{l:non-separated-inside}, we obtain the conclusion of the lemma for all instances $\vec{x}$ with $x_i < x_k$\,:

\begin{itemize}
\item For all instances $\vec{x}$ with $x_i < x_j < x_k$, $x_j \in F(\vec{x})$.

\item For all instances $\vec{x}$ with $x_j < x_i < x_k$, $F_1(\vec{x}) = x_j$ and $F_2(\vec{x}) = x_k$, by Lemma~\ref{l:non-separated-inside}, because $\l_p(j, x_j) =\,\uparrow$. Therefore, $F(\vec{x}) = (\min\vec{x}, \max\vec{x})$.

\item For all instances $\vec{x}$ with $x_i < x_k < x_j$, $F_2(\vec{x}) = x_j$ and $F_1(\vec{x}) = x_i$, by the symmetric version of Lemma~\ref{l:non-separated-inside}, because $\r_p(j, x_j) =\,\downarrow$. Therefore, $F(\vec{x}) = (\min\vec{x}, \max\vec{x})$.
\end{itemize}

As for instances $\vec{x}$ with $x_k < x_i$, we first show that either $\l(k, b) = (j, b)$ or $\l_p(k, b) =\,\uparrow$, for all locations $b$ (i.e., we exclude the possibility that for some $b$, $\l(k, b) = (i, b)$). In the former case, the facilities are allocated as in the case where $x_i < x_k$. In the latter case, the facilities are placed at the two extremes.

%In Case II, where there is an agent-location pair $(i, a)$ with $\l_p(i, a) = a$, we have also to determine the facility allocation for instances $\vec{x}$ where $x_k < x_i$. To this end, we first show that either $\l(k, b) = (j, b)$ or $\l_p(k, b) =\,\uparrow$, for all locations $b$ (i.e., we exclude the possibility that for some location $b$, $\l(k, b) = (i, b)$).

To reach a contradiction, let us assume that for some location $b$, $\l(k, b) = (i, b)$. Let $\vec{y}$ be any $(k,i|j)$-well-separated instance with $y_k = b$. Since $\l(k, b) = (i, b)$, since $k$ is the leftmost agent in $\vec{y}$, and since $y_i > b$, Lemma~\ref{l:non-separated} implies that $y_i \in F(\vec{y})$.
Moreover, since $\r(j, y_j) =\,\downarrow$, and since $j$ is the rightmost and $y_k$ is the leftmost agent in $\vec{y}$, the symmetric version of Lemma~\ref{l:non-separated-inside} implies that $y_k \in F(\vec{y})$. Therefore, the agent $j$ is served by the facility at $y_i$, which contradicts $F$'s bounded approximation ratio, because $\vec{y}$ is an $(k,i|j)$-well-separated instance.

To conclude the proof, we distinguish between the case where there is a location $b$ such that $\l(k, b) = (j, b)$, and the case where for all locations $b$, $\l_p(k, b) =\,\uparrow$.
We recall that $\l_p(j, b) =\,\uparrow$ and $\r_p(j, b) =\,\downarrow$, for all locations $b$, by Lemma~\ref{l:only-dictator}, because $\l(i, a) = (j, a)$.

If there is a location $b$ such that $\l(k, b) = (j, b)$, then for all locations $b' \in \reals$, $\l(k, b') = (j, b')$, by Lemma~\ref{l:always-a}, and $\r(i, b') = (j, b')$, by Lemma~\ref{l:others-a}. Therefore, for all instances $\vec{x}$ with $x_k < x_i$, $x_j \in F(\vec{x})$. More specifically, if $x_k < x_j$, this follows from Lemma~\ref{l:non-separated}, while if $x_j < x_i$, this follows from the symmetric version of Lemma~\ref{l:non-separated}. Then, % as before,
we use Lemma~\ref{l:non-separated-inside}, and obtain the conclusion of the lemma for all instances $\vec{x}$ with $x_k < x_i$\,:

\begin{itemize}
\item For all instances $\vec{x}$ with $x_k < x_j < x_i$, $x_j \in F(\vec{x})$.

\item For all instances $\vec{x}$ with $x_j < x_k < x_i$, $F_1(\vec{x}) = x_j$ and $F_2(\vec{x}) = x_i$, by Lemma~\ref{l:non-separated-inside}, because $\l_p(j, x_j) =\,\uparrow$. Therefore, $F(\vec{x}) = (\min\vec{x}, \max\vec{x})$.

\item For all instances $\vec{x}$ with $x_k < x_i < x_j$, $F_2(\vec{x}) = x_j$ and $F_1(\vec{x}) = x_k$, by the symmetric version of Lemma~\ref{l:non-separated-inside}, because $\r_p(j, x_j) =\,\downarrow$. Therefore, $F(\vec{x}) = (\min\vec{x}, \max\vec{x})$.
\end{itemize}

Otherwise, for all locations $b$, $\l_p(k, b) =\,\uparrow$. We observe that for all locations $b$, $\r_p(i, b) =\,\downarrow$. Otherwise, there would exist a location $b'$ with $\r_p(i, b') = b'$, by Lemma~\ref{l:i-cases}, and thus $\l_p(k, b') = b'$, by Lemma~\ref{l:others-a}, a contradiction.
As before, we now use Lemma~\ref{l:non-separated-inside}, and obtain the conclusion of the lemma for all instances $\vec{x}$ with $x_k < x_i$. More specifically, since $x_k < x_i$, the leftmost agent of $\vec{x}$ is either $k$ or $j$. If the leftmost agent is $k$ (resp. $j$), $F_2(\vec{x}) = \max\vec{x}$, by Lemma~\ref{l:non-separated-inside}, because $\l_p(k, x_k) =\,\uparrow$ (resp. $\l_p(j, x_j) =\,\uparrow$).
Similarly, the rightmost agent of $\vec{x}$ is either $j$ or $i$. If the rightmost agent is $j$ (resp. $i$), $F_1(\vec{x}) = \min\vec{x}$, by the symmetric version of Lemma~\ref{l:non-separated-inside}, because $\r_p(j, x_j) =\,\downarrow$ (resp.
$\r_p(i, x_i) =\,\downarrow$).
Thus, if $\l_p(k, b) =\,\uparrow$, for all instances $\vec{x}$ with $x_k < x_i$, $F(\vec{x}) = (\min\vec{x}, \max\vec{x})$.
\qed\end{proof}

\section{Strategyproof Allocation of 2 Facilities to 3 Locations: The Proof of Theorem~\ref{thm:3locations}}
\label{s:3locations}

The proof of Theorem~\ref{thm:3locations} is based on the following extension of Theorem~\ref{thm:3agents} to 3-location instances. In fact, we can restate the whole proof of Theorem~\ref{thm:3agents} with 3 coalitions of agents instead of 3 agents. Then, using that any strategyproof mechanism is also partial group strategyproof \cite[Lemma~2.1]{LSWZ10}, we obtain that:

\begin{corollary}\label{cor:3agents}
Let $F$ be any nice mechanism applied to 3-location instances with $n \geq 3$ agents. Then, there exist at most two permutations $\pi_1$, $\pi_2$ of the agent coalitions with $\pi_1(2) = \pi_2(2)$ such that for all instances $\vec{x}$ where the coalitions are arranged on the line according to $\pi_1$ or $\pi_2$, $F(\vec{x})$ places a facility at the location of the middle coalition. For any other permutation $\pi$ and instance $\vec{x}$, where the agent coalitions are arranged on the line according to $\pi$, $F(\vec{x}) = ( \min \vec{x}, \max \vec{x} )$.
\end{corollary}

A central notion in the proof of Theorem~\ref{thm:3locations} is that of a \emph{dictator coalition}. A (non-empty) coalition $C \subset N$, $|C| \leq |N|-2$, is called a dictator for 3-location instances, if for all partitions $N_1, N_2$ of $N \setminus C$ and all instances $\vec{x} = (x_1\sep N_1, x\sep C, x_2\sep N_2) \in \Ithr(N)$, $x \in F(\vec{x})$.
The first key step of the proof is to show that if there exists a 3-location instance where the middle coalition has at most $n-3$ agents and is allocated a facility, then any superset of the middle coalition is a dictator for 3-location instances.

\begin{lemma}\label{l:any-partition}
Let $N$ be a set of $n \geq 4$ agents. If there exists a 3-location instance $\vec{x} = (x_1\sep N_1, x_2\sep N_2, x_3\sep N_3)$ with $x_1 < x_2 < x_3$ and $|N_2| \leq n-3$, such that $x_2 \in F(\vec{x})$, then any coalition $N'_2 \supseteq N_2$ is a dictator for 3-location instances.
\end{lemma}

\begin{proof}
%Let $\vec{x} = (x_1\sep N_1, x_2\sep N_2, x_3\sep N_3)$ be a $3$-location instance such that $x_1 < x_2 < x_3$, $|N_2| \leq n-3$, and $x_2 \in F(\vec{x})$.
%
We consider any $3$-location instance $\vec{y} = (y_1\sep N'_1, y_2\sep N'_2, y_3\sep N'_3)$ with $N'_2 \supseteq N_2$, and show that $y_2 \in F(\vec{y})$.
We start with the case where $y_1 < y_3$. By the hypothesis about the existence of $\vec{x}$, and by Corollary~\ref{cor:3agents}, $F$ allocates a facility to the coalition $N_2$ for any instance $(y_1\sep N_1, y_2\sep N_2, y_3\sep N_3)$ with $y_1 < y_3$. In particular, the existence of $\vec{x}$, where the middle coalition $N_2$ is allocated a facility, implies that $N_2$ serves as the partial dictator of Corollary~\ref{cor:3agents}. Therefore, $F$ allocates a facility to $N_2$ if either $N_2$ is the middle coalition and $y_1 < y_3$, or it is the left or the right coalition.

Therefore, we only need to show that $F$ allocates a facility to the middle coalition $N'_2$ if $N_2 \subset N'_2$, and either $N'_1 \neq N_1$, or $N'_3 \neq N_3$ (or both). To this end, we show how to move agents between the left and the right coalition and from the left and the right coalitions to the middle coalition, and obtain the desired partition $N'_1, N'_2, N'_3$ of agents, while keeping a facility of $F$ allocated to the middle coalition. Then, the lemma follows from Corollary~\ref{cor:3agents}.

For simplicity, we assume that $|N_1| \geq 2$ and that $N_1 \cut N'_1 \neq \emptyset$. We show how to remove these assumptions later on.
We first consider the instance $\vec{y}' = (y_1\sep N_1, y_1+\eps\sep N_2, y_3\sep N_3)$, where $\eps > 0$ is chosen small enough that due to $F$'s bounded approximation ratio, the rightmost facility of $F(\vec{y}')$ is placed on the right of $y_1+\eps$. Since $F(\vec{x})$ allocates a facility to $N_2$,  Corollary~\ref{cor:3agents} implies that $y_1+\eps \in F(\vec{y}')$.
We observe that as long as there are at least two agents at $y_1$, we can move either some agent $j \in N_1 \cut N'_2$ from $y_1$ to $y_1+\eps$ or some agent $j \in N_1 \cut N'_3$ from $y_1$ to $y_3$, while keeping a facility of $F$ to $y_1+\eps$.
More specifically, if an agent $j \in N_1 \cut N'_2$ moves from $y_1$ to $y_1+\eps$, it holds that $y_1+\eps \in F(y_1\sep N_1 \setminus \{ j \}, y_1+\eps\sep N_2 \union \{ j \}, y_3\sep N_3)$, due to $F$'s strategyproofness. Otherwise, agent $j$ in instance $(y_1\sep N_1 \setminus \{ j \}, y_1+\eps\sep N_2 \union \{j\}, y_3\sep N_3)$ could manipulate $F$ by reporting $y_1$ instead of $y_1+\eps$.
If an agent $j \in N_1 \cut N'_3$ moves from $y_1$ to $y_3$, it holds that $y_1+\eps \in F(y_1\sep N_1 \setminus \{ j \}, y_1+\eps\sep N_2, y_3\sep N_3 \union \{j\})$, by Corollary~\ref{cor:3agents}. In particular, if the middle coalition is not allocated a facility in $(y_1\sep N_1 \setminus \{ j \}, y_1+\eps\sep N_2, y_3\sep N_3 \union y_3\sep N_3 \union \{j\})$, Corollary~\ref{cor:3agents} implies that the two facilities are placed at $y_1$ and $y_3$. Therefore, agent $j$ in instance $(y_1\sep N_1, y_1+\eps\sep N_2, y_3\sep N_3)$ could manipulate $F$ by reporting $y_3$ instead of $y_1$.
By repeatedly applying this argument, we move all agents in $N_1 \cut N'_2$ from the left to the middle coalition and all agents in $N_1 \cut N'_3$ from the left to the right coalition, and keep at least one agent at $y_1$, since $N_1 \cut N'_1 \neq \emptyset$. Thus, we obtain an instance $\vec{z} = (y_1\sep Z_1, y_1+\eps\sep Z_2, y_3\sep Z_3)$ such that $Z_1 = N_1 \cut N'_1 \neq \emptyset$, $Z_2 = N_2 \union (N_1 \cut N'_2)$, $Z_3 = N_3 \union (N_1 \cut N'_3)$, and $y_1+\eps \in F(\vec{z})$. Moreover, since $F$ allocates a facility to the middle coalition $Z_2$ of $\vec{z}$, by Corollary~\ref{cor:3agents}, $F$ allocates a facility to the coalition $Z_2$ for all instances $(y_1\sep Z_1, y\sep Z_2, y_3\sep Z_3)$ with $y_1 < y_3$.

Next, we consider the instance $\vec{z}' = (y_1\sep Z_1, y_3-\eps'\sep Z_2, y_3\sep Z_3)$, where $\eps' > 0$ is chosen small enough that due to $F$'s bounded approximation ratio, the leftmost facility of $F(\vec{z}')$ is located on the left of $y_3-\eps'$. By Corollary~\ref{cor:3agents}, $y_3-\eps' \in F(\vec{z}')$.
By an argument symmetric to the argument above, we obtain that as long as there are at least two agents at $y_3$, we can move either some agent $j \in Z_3 \cut N'_2$ from $y_3$ to $y_3-\eps$ or some agent $j \in Z_3 \cut N'_1$ from $y_3$ to $y_1$, while keeping a facility of $F$ to $y_3-\eps'$.
As before, by repeatedly applying this argument, we move all agents in $Z_3 \cut N'_2$ from the right to the middle coalition and all agents in $Z_3 \cut N'_1$ from the right to the left coalition. Thus, we obtain the instance $\vec{y}'' = (y_1\sep N'_1, y_3-\eps\sep N'_2, y_3\sep N'_3)$ where $y_3-\eps \in F(\vec{y}'')$.
Since $F$ allocates a facility to the middle coalition $N'_2$ of $\vec{y}''$, by Corollary~\ref{cor:3agents}, $F$ allocates a facility to the coalition $N'_2$ for all instances $\vec{y} = (y_1\sep N'_1, y_2\sep N'_2, y_3\sep N'_3)$ with $y_1 < y_3$.

We are now ready to remove the assumptions that $|N_1| \geq 2$ and that $N_1 \cut N'_1 \neq \emptyset$. If $|N_1| = 1$, then $|N_2| \geq 2$, because $|N_2| \leq n-3$, and we start moving agents from $N_3$ (i.e., we first consider the instance $\vec{z}'$ and then the instance $\vec{y}'$). If $N_1 \cut N'_1 = \emptyset$, we first apply two rounds of agent moves between the left and the right coalition, so that every agent in $N'_1 \union (N'_2 \setminus N_2)$ ends up in the left coalition and every agent in $N'_2$ ends up in the right coalition. In a final half-round of agent moves, where we consider only the instance $\vec{y}'$, every agent in $N'_2 \setminus N_2$ moves from the left coalition to the middle coalition. This also deals with the case where $|N'_2| = n-2$.
If $y_1 > y_3$, we first work as above and move all agents in $N_1$ from $y_1$ to $y_3$ and all agents in $N_3$ from $y_3$ and $y_1$. Thus, we obtain an instance $(y_3\sep N_1, y_2\sep N_2, y_1\sep N_3)$, with $y_3 < y_1$, and proceed as above.
\qed\end{proof}

To conclude the proof of Theorem~\ref{thm:3locations}, we distinguish, for technical reasons, between instances where the largest coalition has size at most $n-3$, and instances where the largest coalition has size $n-2$ and the other two coalitions are singletons. Given a set $N$ of agents, we let $\IthrS(N)$ denote the former class and $\IthrL(N)$ denote the latter class of 3-location instances. We note that for all sets $N$ of $n \geq 5$ agents, $\IthrS(N)$ and $\IthrL(N)$ form a partition of $\Ithr(N)$. The following pair of lemmas establish Theorem~\ref{thm:3locations} first for instances in $\IthrS(N)$, and then for all instances in $\Ithr(N)$. We first establish Theorem~\ref{thm:3locations} for instances in $\IthrS(N)$.

\begin{lemma}\label{l:unique-dictator-small}
Let $N$ be a set of $n \geq 5$ agents. If there is an instance $\vec{x} = (x_1\sep N_1, x_2\sep N_2, x_3\sep N_3) \in \IthrS(N)$ with $x_1 < x_2 < x_3$, such that $x_2 \in F(\vec{x})$, then there exists a unique agent $j \in N_2$ such that for all instances $\vec{y} \in \IthrS(N)$, $y_j \in F(\vec{y})$.
\end{lemma}

\begin{proof}
Let $\vec{x} = (x_1\sep N_1, x_2\sep N_2, x_3\sep N_3) \in \IthrS(N)$ with $x_1 < x_2 < x_3$, such that $x_2 \in F(\vec{x})$. Then, by Lemma~\ref{l:any-partition}, $N_2$ is a dictator for 3-location instances. For sake of contradiction, we assume that there is a minimal (sub)coalition $N'_2 \subseteq N_2$, $|N'_2| \geq 2$, that violates the lemma. Namely, $N'_2$ is a dictator for 3-location instances, while for every agent $j \in N'_2$, $N'_2 \setminus \{ j \}$ is not a dictator.
The lemma follows from the observation that if such a minimal dictator coalition $N'_2$ exists, then, for any agent $j \in N'_2$ and any agent $i \in N \setminus N'_2$, $\{ j, i \}$ is a dictator for 3-location instances.

Before proving this claim, let us first show that it indeed implies the lemma. By the claim above, if we let $j_1, j_2$ be any two agents in the minimal dictator coalition $N'_2$, for any pair of agents $i_1, i_2 \in N \setminus N'_2$, the coalitions $\{ j_1, i_1 \}$ and $\{ j_2, i_2 \}$ are both dictators for 3-location instances.
But the existence of two disjoint dictator coalitions contradicts the hypothesis that $F$ has a bounded approximation ratio.
To see this, we consider the instance $\vec{z} = (0\sep  N \setminus \{j_1, i_1, j_2, i_2\}, 1\sep \{ j_1, i_1\}, 1+\eps\sep \{j_2, i_2\})$, where $\eps > 0$ is chosen sufficiently small. Since both $\{ j_1, i_1\}$ and $\{ j_2, i_2 \}$ are dictators for 3-location instances, $F(\vec{z}) = (1, 1+\eps)$, which for a sufficiently small $\eps$, contradicts that $F$ has a bounded approximation ratio.

Therefore, there exists an agent $j \in N_2$ such that for all instances $\vec{y} \in \IthrS(N)$, $y_j \in F(\vec{y})$. The uniqueness of such an agent $j$ follows from an argument similar to the argument above, due to $F$'s bounded approximation ratio.

To complete the proof of Lemma~\ref{l:unique-dictator-small}, we have also to show that if there is a minimal coalition $N'_2$, $|N'_2| \geq 2$, such that $N'_2$ is a dictator for $3$-location instances, while for every agent $j \in N'_2$, $N'_2 \setminus \{ j \}$ is not a dictator for $3$-location instances, then for any agent $j \in N'_2$ and any agent $i \in N \setminus N'_2$, the coalition $\{ j, i \}$ is a dictator for $3$-location instances.
For sake of contradiction, we assume that there is an agent $j \in N'_2$ and an agent $i \in N \setminus N'_2$, such that $\{ i, j \}$ is not a dictator for $3$-location instances. For simplicity of notation, we let $C' = N'_2 \setminus \{ j \}$, $C_j = \{j, i\}$, and $N' = N \setminus (C' \union C_j)$. Since the coalition $C_j$ is not a dictator, Lemma~\ref{l:any-partition} implies that for all instances $\vec{x} = (x_1\sep N', x_2\sep C_j, x_3\sep C')$ with $x_1 < x_2 < x_3$, $x_2 \not\in F(\vec{x})$.
To reach a contradiction, we consider any such instance $\vec{x}$, and choose a location $r > 2|x_3|+|x_2|$ large enough that for the ($4$-location) instance $\vec{x}' = (\vec{x}_{-i}, r)$, $r \in F(\vec{x}')$. Such an $r$ exists because $F$ has a bounded approximation ratio, and thus every hole in the image set $I_i(\vec{x}_{-i})$ is a bounded interval.

Let $a$ be the location of the other facility of $F(\vec{x}')$. We show that there is no choice of $a$ compatible with the assumption that $F$ is strategyproof, thus obtaining a contradiction.
More specifically, if $a = x_2$, the agent $i$ in the instance $\vec{x}$ could manipulate $F$ by reporting $r$ instead of $x_2$.
If $a \in (x_2, +\infty)$, the coalition $N'$ in $\vec{x}$ could manipulate $F$ by reporting $x_2$ instead of $x_1$. Then, $\vec{x}_1 = (x_2\sep N' \union \{j\}, x_3\sep C', r\sep \{i\})$ is a $3$-location instance, and since the coalition $C'$ is not a dictator for $3$-location instances, $F_1(\vec{x}_1) = x_2$. Otherwise, by Corollary~\ref{cor:3agents}, $x_3 \in F(\vec{x})$, and thus, by Lemma~\ref{l:any-partition}, $C'$ would be a dictator.
If $a \in [-\infty, x_2)$, the coalition $C'$ in $\vec{x}$ could manipulate $F$ by reporting $x_2$ instead of $x_3$. Then, $\vec{x}_2 = (x_1\sep N', x_2\sep N'_2, r\sep \{i\})$ is a $3$-location instance, and since the coalition $N'_2$ is a dictator for $3$-location instances, $x_2 \in F(\vec{x}_2)$.

Therefore, there is an instance $\vec{x} = (x_1\sep N', x_2\sep C_j, x_3\sep C')$, with $x_1 < x_2 < x_3$ and $|C_j| \leq n-3$, such that $x_2 \in F(\vec{x})$. Thus, by Lemma~\ref{l:any-partition}, the coalition $C_j$ is a dictator for $3$-location instances. This concludes the proof of the claim and the proof of the lemma.
\qed\end{proof}

The next lemma shows that that $F$ behaves in the same way for all instances in $\Ithr(N)$, and concludes the proof of Theorem~\ref{thm:3locations}. Interestingly, Lemma~\ref{l:unique-dictator-large} shows that we can tell whether all instances in $\Ithr(N)$ admit a dictator or not, by only checking whether instances in $\IthrS(N)$ admit a dictator.

%The proof of the following lemma is similar to the proof of Lemma~\ref{l:unique-dictator-small}, and can be found in the Appendix, Section~\ref{app:s:unique-dictator-large}.

\begin{lemma}\label{l:unique-dictator-large}
Let $N$ be a set of $n \geq 5$ agents. If there is an instance $\vec{x} = (x_1\sep N_1, x_2\sep N_2, x_3\sep N_3) \in \IthrS(N)$ with $x_1 < x_2 < x_3$, such that $x_2 \in F(\vec{x})$, then there exists a unique agent $j \in N_2$ such that for all instances $\vec{y} \in \Ithr(N)$, $y_j \in F(\vec{y})$. Otherwise, for all instances $\vec{y} \in \Ithr(N)$, $F(\vec{y}) = (\min \vec{y}, \max \vec{y})$.
\end{lemma}

\begin{proof}
We distinguish between the case where for some $\vec{x} = (x_1\sep N_1, x_2\sep N_2, x_3\sep N_3) \in \IthrS(N)$, with $x_1 < x_2 < x_3$, $x_2 \in F(\vec{x})$, and the case where for all instances $\vec{y} \in \IthrS(N)$, $F(\vec{y}) = (\min\vec{y}, \max\vec{y})$.

In the former case, Lemma~\ref{l:unique-dictator-small} implies the existence of a unique agent $j \in N_2$ such that for all instances $\vec{y} \in \IthrS(N)$, $y_j \in F(\vec{y})$. Let $i, k \in N$ be any two agents different from $j$. Since the instance $\vec{x}' = (x_1\sep N\setminus\{i,j,k\}, x_2\sep \{ j \}, x_3\sep \{i, k\})$ is a $3$-location instance in $\IthrS(N)$, $x_2 \in F(\vec{x}')$. Moreover, since $x_1 < x_2 < x_3$ and the cardinality of the middle coalition of $\vec{x}'$ is at most $n-3$, Lemma~\ref{l:any-partition} implies that any coalition $N'_2$ that includes $j$ is a dictator for $3$-location instances. Therefore, for all instances $\vec{y} \in \Ithr(N)$, $y_j \in F(\vec{y})$. The uniqueness of such an agent $j$ follows from the bounded approximation ratio of $F$, as in the proof of Lemma~\ref{l:unique-dictator-small}.

In the latter case, for all instances $\vec{y} \in \IthrS(N)$, $F(\vec{y}) = (\min \vec{y}, \max \vec{y})$, and thus instances in $\IthrS(N)$ do not admit a dictator. We next show that for all $\vec{y} \in \IthrL(N)$, it is also the case that $F(\vec{y}) = (\min \vec{y}, \max \vec{y})$.
For sake of contradiction, let us assume that there is an instance $\vec{z} = (z_1\sep N_1, z_2\sep N_2, z_3\sep N_3) \in \IthrL(N)$ with $z_1 < z_2 < z_3$, such that $F(\vec{z}) \neq (z_1, z_3)$. Thus, by Corollary~\ref{cor:3agents}, $z_2 \in F(\vec{z})$. Since $\vec{z} \in \IthrL(N)$, a coalition has size $n-2$ and the other two coalitions are singletons. If the middle coalition $N_2$ is a singleton, by Lemma~\ref{l:any-partition}, the agent in $N_2$ is a dictator for all $3$-location instances. Otherwise, $|N_2| = n-2$, and Claim~\ref{cl:remove-agent-claim} implies that for any agent $j \in N_2$, the coalition $N_3 \union \{ j \}$ is a dictator for $3$-location instances. In both cases, we reach a contradiction to the hypothesis that instances in $\IthrS(N)$ do not admit a dictator.

\begin{myclaim}\label{cl:remove-agent-claim}
Let $N$ be a set of $n \geq 5$ agents. If for all instances $\vec{y} \in \IthrS(N)$, $F(\vec{y}) = (\min\vec{y}, \max\vec{y})$, and there exists an instance $\vec{z} = (z_1\sep N_1, z_2\sep N_2, z_3\sep N_3)$, with $z_1 < z_2 < z_3$ and $N_2 = |n-2|$, such that $z_2 \in F(\vec{z})$, then for any agent $j \in N_2$, the coalition $N_3 \union \{ j \}$ is a dictator for $3$-location instances.
\end{myclaim}

\begin{proof}[of Claim~\ref{cl:remove-agent-claim}]
%The proof is similar to the proof in Section~\ref{app:s:unique-dictator-small}. %
For simplicity of notation, we let $j \in N_2$ be any agent, let $C' = N_2 \setminus \{ j \}$, let $i$ be the unique agent in $N_1$ and $k$ be the unique agent in $N_3$, and let $C_j = \{ j, k \}$.
For sake of contradiction, let us assume that the coalition $C_j$ is not a dictator for $3$-location instances. Therefore, since $|C_j| \leq n-3$, by Lemma~\ref{l:any-partition}, for all instances $\vec{x} = (x_1\sep \{i\}, x_2\sep C_j, x_3\sep C')$, with $x_1 < x_2 < x_3$, $x_2 \not\in F(\vec{x})$.
To reach a contradiction, we consider any such instance $\vec{x}$, and choose a location $r > 2|x_3|+|x_2|$ large enough that for the ($4$-location) instance $\vec{x}' = (\vec{x}_{-k}, r)$, $r \in F(\vec{x}')$. Such an $r$ exists because $F$ has a bounded approximation ratio, and thus every hole in the image set $I_k(\vec{x}_{-k})$ is a bounded interval.

Let $a$ be the location of the other facility of $F(\vec{x}')$. We show that there is no choice of $a$ compatible with the assumption that $F$ is strategyproof, thus obtaining a contradiction.
More specifically, if $a = x_2$, the agent $k$ in the instance $\vec{x}$ could manipulate $F$ by reporting $r$ instead of $x_2$.
If $a \in (x_2, +\infty)$, the agent $i$ in $\vec{x}$ could manipulate $F$ by reporting $x_2$ instead of $x_1$. Then, $\vec{x}_1 = (x_2\sep \{i, j\}, x_3\sep C', r\sep \{k\})$ is a $3$-location instance in $\IthrS(N)$, and by the hypothesis of the claim, $F(\vec{x}_1) = (x_2, r)$.
If $a \in [-\infty, x_2)$, the coalition $C'$ in $\vec{x}$ could manipulate $F$ by reporting $x_2$ instead of $x_3$. Then, $\vec{x}_2 = (x_1\sep \{i\}, x_2\sep N_2, r\sep \{k\})$ is a $3$-location instance where $N_1 = \{i\}$ is the left coalition, $N_2$ is the middle coalition, and $N_3 = \{k\}$ is the right coalition. Therefore, by the hypothesis of the claim and Corollary~\ref{cor:3agents}, $x_2 \in F(\vec{x}_2)$.

Hence, there is an instance $\vec{x} = (x_1\sep \{i\}, x_2\sep C_j, x_3\sep C')$, with $x_1 < x_2 < x_3$ and $|C_j| \leq n-3$, such that $x_2 \in F(\vec{x})$. Thus, by Lemma~\ref{l:any-partition}, the coalition $C_j$ is a dictator for $3$-location instances. \qed\end{proof}

With the proof of Claim~\ref{cl:remove-agent-claim}, we conclude the proof of the lemma.
\qed\end{proof}

\section{Strategyproof Allocation of 2 Facilities: The Proof of Theorem~\ref{thm:2fac-gen}}
\label{s:2fac-gen}

The final step is that we extend Theorem~\ref{thm:3locations} to general instances with $n \geq 5$ agents, and conclude the proof of Theorem~\ref{thm:2fac-gen}. %In particular, we show that any nice mechanism applied to instances with $n \geq 5$ agents, either admits a dictator, or it always allocates the facilities to the two extreme locations of the instance.
The proof considers two different cases, depending on how the mechanism $F$ behaves for 3-location instances, and proceeds by induction on the number of different locations.

We first consider the case where $F$ admits a dictator $j$ for 3-location instances, and show that agent $j$ is a dictator for all $\vec{x} \in \I(N)$.
For sake of contraction, we assume %that there is
an instance $\vec{x} = (x_1, \ldots, x_n) \in \I(N)$ for which $x_j \not\in F(\vec{x})$. %Without loss of generality,
W.l.o.g., we let $k \neq j$ be the rightmost agent of $\vec{x}$ (if $j$ is the rightmost agent, the argument is symmetric). Since $x_j \not\in F(\vec{x})$, there is a $x_j$-hole $(l, r)$ in the imageset $I_j(\vec{x}_{-j})$.
For a small $\eps \in (0, (r-l)/2)$, we consider the instance $\vec{x}_1 = (\vec{x}_{-j}, l+\eps)$, where $j$ moves from $x_j$ to $l+\eps$. By strategyproofness, and since $l+\eps$ is in the left half of the hole $(l, r)$, $l \in F(\vec{x}_1)$.
Then, we iteratively move all agents $i \in N \setminus \{j, k\}$ from $x_i$ to $l$. By strategyproofness, if $F$ has a facility at $l$ before $i$ moves from $x_i$ to $l$, $F$ keeps its facility at $l$ after $i$'s move. Otherwise, agent $i$ with location $l$ could manipulate $F$ by reporting $x_i$.
Thus, we obtain a 3-location instance $\vec{x}' = (l\sep N \setminus \{j, k\}, l+\eps\sep \{j\}, x_k\sep \{k\})$ with $l < l+\eps < x_k$, such that $l \in F(\vec{x}')$. Moreover, since $j$ is a dictator for 3-location instances, $l+\eps \in F(\vec{x}')$, and thus $F(\vec{x}') = (l, l+\eps)$. For $\eps$ sufficiently smaller than $x_k - l$, this contradicts the bounded approximation ration of $F$.

%The case where $F$ does not admit a dictator for 3-location instances is similar. 
If $F$ does not admit a dictator for 3-location instances, by Theorem~\ref{thm:3locations}, $F(\vec{x}) = (\min\vec{x}, \max\vec{x})$ for all instances $\vec{x} \in \Ithr(N)$. Next, we show that in this case, $F(\vec{x}) = (\min\vec{x}, \max\vec{x})$ for all instances $\vec{x} \in \I(N)$.
For sake of contradiction, we assume that there exists some instance $\vec{x} = (x_1, \ldots, x_n) \in \I(N)$, for which $F(\vec{x}) \neq (\min\vec{x}, \max\vec{x})$. We let $j$ be the leftmost and $k$ be the rightmost agent of $\vec{x}$. Since $F(\vec{x}) \neq (x_j, x_k)$, there is a location $a \in F(\vec{x})$ with $a \neq x_j$ and $a \neq x_k$.

If $x_j < a < x_k$, we iteratively move all agents $i \in N \setminus \{j, k\}$ from $x_i$ to $a$. As in the previous case, since $F$ is strategyproof, it keeps allocating a facility at $a$ after each agent $i$ moves to $a$. Thus, we obtain a 3-location instance $\vec{x}' = (x_j\sep \{j\}, a\sep N \setminus \{j, k\}, x_k\sep \{k\})$ for which $F$ does not allocate the facilities to the two extremes, a contradiction. %Therefore, for all instances $\vec{x} \in \I(N)$, $F(\vec{x})$ does not allocate a facility inside the two extremes.

We proceed to consider the case where $a < x_j$, %i.e., where $a$ lies outside the two extreme locations
(the case where $a > x_k$ is identical). %Let $\ell \geq 4$ be the number of distinct locations occupied by the agents in $\vec{x}$.
Without loss of generality, we assume that the selected instance $\vec{x}$ has the minimum number of distinct locations among all instances for which $F$ allocates a facility outside the two extremes. %Namely, we assume that for all instances $\vec{x}' \in \I(N)$ with at most $\ell-1$ distinct locations, $\min\vec{x}' \leq F_1(\vec{x}') \leq F_2(\vec{x}') \leq \max\vec{x}'$, and by what we have shown above, $F(\vec{x}') = (\min\vec{x}', \max\vec{x}')$.
Since $a < x_j$, either $x_j$ or $x_k$ is not allocated a facility by $F(\vec{x})$. Next, we assume that $x_j \not\in F(\vec{x})$ (the case where $x_k \not\in F(\vec{x})$ is symmetric). Let $S_j \subseteq N$ be the set of agents located at $x_j$, and let $b = \min\vec{x}_{-S_j}$ be the second location from the left in $\vec{x}$.
%We note that $S_j \neq \emptyset$, since $j \in S_j$.
Since $x_j \not\in F(\vec{x})$, there is a $x_j$-hole $(l, r)$ in the image set $I_{S_j}(\vec{x}_{-S_j})$.
We observe that $r \leq b$, because if all agents in $S_j$ move from $x_j$ to $b$, we obtain the instance $\vec{x}' = (\vec{x}_{-S_j}, (b, \ldots, b))$ that has less distinct locations than $\vec{x}$ and $b$ as its leftmost location. Since $\vec{x}$ has the minimum number of distinct locations among all instances for which $F$ allocates a facility outside the two extremes, $F(\vec{x}')$ allocates a facility to $b$. %and thus $b \in I_{S_j}(\vec{x}_{-S_j})$.
We now choose $\eps > 0$ such that $r-\eps$ lies in the right half of the hole $(l, r)$, and move all agents in $S_j$ from $x_j$ to $r-\eps$. Thus, we obtain the instance $\vec{x}'' = (\vec{x}_{-S_j}, (r-\eps, \ldots, r-\eps))$. Since $F$ is strategyproof and $r$ is the closest location to $r-\eps$ in $I_{S_j}(\vec{x}_{-S_j})$, $F(\vec{x}'')$ allocates a facility to $r > r-\eps$ (see also \cite[Lemma~3.1]{LSWZ10}).
Therefore, $F(\vec{x}'')$ allocates a facility inside the two extremes of $\vec{x}''$, which contradicts what we have shown above: namely that if $F$ does not admit a dictator for 3-location instances, then $F$ never allocates a facility inside the two extremes. \qed

\section{Inexistence of Anonymous Nice Mechanisms for More Than 2 Facilities}
\label{s:3facilities}

We next obtain an impossibility result for anomynous nice $K$-Facility Location mechanisms, for all $K \geq 3$.

%We next show that for all $K \geq 3$, there are no anonymous nice mechanisms for $K$-Facility Location with $n \geq K+1$ agents.

\begin{theorem}\label{thm:3facilities}
For every $K \geq 3$, any deterministic anonymous strategyproof mechanism for $K$-Facility Location with $n \geq K+1$ agents on the real line has an unbounded approximation ratio.
\end{theorem}

\begin{proof}
We only consider the case where $K = 3$ and $n = 4$. It is straightforward to verify that the proof generalizes to any $K \geq 3$ and any $n \geq K+1$.
For sake of contradiction, we let $F$ be an anonymous nice mechanism for $3$-Facility Location, and let $\rho$ be the approximation ratio of $F$ for instances with $4$ agents. Next, we construct a family of instances for which the approximation ratio of $F$ is greater than $\rho$.

Since $F$ is anonymous, we assume that for any instance $\vec{x}$, $x_1 < x_2 < x_3 < x_4$ (i.e. the agents are arranged on the line in increasing order of their indices).
For some sufficiently large $\lambda > \rho$, we consider the instance $\vec{x} = (0, \lambda, 3\lambda^2+\lambda, 3\lambda^2+\lambda+1)$, which is $(1|2|3,4)$-well-separated. By Proposition~\ref{prop:interval}, $F_3(\vec{x}) \in [x_3, x_4]$. W.l.o.g., we assume that $F_3(\vec{x}) \in \{x_3, x_4\}$. Otherwise, if $F_3(\vec{x}) = a$ and $x_3 < a < x_4$, the instance $(\vec{x}_{-4}, a)$ is also $(1|2|3,4)$-well-separated and has $F_3(\vec{x}_{-4}, a) = a$, due to $F$'s strategyproofness.

We start with the case where $F_3(\vec{x}) = x_4$. Since $\vec{x}$ is an $(1|2|3,4)$-well-separated instance, both $x_3$ and $x_4$ are served by the facility at $x_4$. Hence, there is a $x_3$-hole $(l, r)$ in the image set $I_3(\vec{x}_{-3})$.
We note that $3\lambda^2+\lambda = x_3 < r \leq x_4 = 3\lambda^2+\lambda+1$, since $x_3 \not\in F(\vec{x})$ and $x_4 \in F(\vec{x})$, and that $l \geq \lambda^2+\lambda-1$.
As for the latter, if $l < \lambda^2+\lambda-1$, then $y = 2\lambda^2+\lambda$ would lie in the right half of the hole $(l, r)$. Thus, if agent $3$ moves to $y$, by strategyproofness, the nearest facility to $y$ in $F(\vec{x}_{-3}, y)$ would be at $r > 3\lambda^2+\lambda$, and thus $\cost[F(\vec{x}_{-3}, y)] > \lambda^2$. Since the optimal cost for $(\vec{x}_{-3}, y)$ is $\lambda$, $F$'s approximation ratio would be $\lambda > \rho$.

Let us now consider the instance $\vec{x}' = (\vec{x}_{-3}, l+\eps)$, where $\eps \in (0, 1]$ is chosen small enough that $l+\eps$ lies in the left half of the hole $(l, r)$ and the instance $(0, \lambda, l, l+\eps)$ is $(1|2|3,4)$-well-separated. Since $F$ is strategyproof, and since $l$ is the nearest point to $l+\eps$ in $I_3(\vec{x}_{-3})$, $l \in F(\vec{x}')$.
Then, we consider the instance $\vec{x}'' = (\vec{x}'_{-4}, l)$. Since $F$ is anonymous%
\footnote{We highlight that the agents $3$ and $4$ implicitly switch indices in $\vec{x}'$ and $\vec{x}''$. More specifically, since we require that the agents are arranged on the line in increasing order of their indices, the location of agent $3$ is $l+\eps$ in $\vec{x}'$ and $l$ in $\vec{x}''$, and the location of agent $4$ is $3\lambda^2+\lambda+1$ in $\vec{x}'$ and $l+\eps$ in $\vec{x}''$. Therefore, to argue about the outcome of $F(\vec{x}'')$ based on the outcome of $F(\vec{x}')$, we resort to the anonymity of $F$.}
and strategyproof, and since $l \in F(\vec{x}')$, $x''_3 = l \in F(\vec{x}'')$.
Moreover, by Proposition~\ref{prop:left_cover}, $x''_4 = l+\eps \in F(\vec{x}'')$, because for the $(1|2|3,4)$-well-separated instance $\vec{x}$, $F_3(\vec{x}) = x_4$, and $\vec{x}''$ is an $(1|2|3,4)$-well-separated instance with $x''_4 \leq x_4$.
Since both $x''_3, x''_4 \in F(\vec{x}'')$, either the agents $1$ and $2$ are served by the same facility of $F(\vec{x}'')$ or the agent $2$ is served by the facility at $l$. In both cases, $\cost[F(\vec{x}'')] \geq \lambda$. But the optimal cost for $\vec{x}''$ is $\eps \leq 1$, and $F$'s approximation ratio is at least $\lambda > \rho$.

Next, we consider the case where $F_3(\vec{x}) = x_3$, which is symmetric to the case where $F_3(\vec{x}) = x_4$. As before, both $x_3$ and $x_4$ are served by the facility at $x_3$, and there is a $x_4$-hole $(l, r)$ in the image set $I_4(\vec{x}_{-4})$.
We note that $3\lambda^2+\lambda = x_3 \leq l < x_4 = 3\lambda^2+\lambda+1$, since $x_3 \in F(\vec{x})$ and $x_4 \not\in F(\vec{x})$, and that $r \leq 5\lambda^2+\lambda+2$.
As for the latter, if $r > 5\lambda^2+\lambda+2$, then $y = 4\lambda^2+\lambda+1$ would lie in the left half of the hole $(l, r)$. Therefore, if agent $4$ moves to $y$, by $F$'s strategyproofness, the nearest facility to $y$ in $F(\vec{x}_{-4}, y)$ would be at $l < 3\lambda^2+\lambda+1$, and thus $\cost[F(\vec{x}_{-4}, y)] > \lambda^2$. Since the optimal cost for instance $(\vec{x}_{-4}, y)$ is $\lambda$, the approximation ratio of $F$ would be $\lambda > \rho$.

Let us now consider the instance $\vec{x}' = (\vec{x}_{-4}, r-\eps)$, where $\eps \in (0, 1]$ is chosen small enough that $r-\eps$ lies in the right half of the hole $(l, r)$ and the instance $(0, \lambda, r-\eps, r)$ is $(1|2|3,4)$-well-separated. Since $F$ is strategyproof, and since $r$ is the nearest point in $I_4(\vec{x}_{-4})$ to $r-\eps$, $r \in F(\vec{x}')$.
Then, we consider the instance $\vec{x}'' = (\vec{x}'_{-3}, r)$ (as before, since we require that the agents are arranged on the line in increasing order of their indices, the agents $3$ and $4$ switch indices in $\vec{x}'$ and $\vec{x}''$).
Since $F$ is anonymous and strategyproof, and since $r \in F(\vec{x}')$, $x''_4 = r \in F(\vec{x}'')$. Moreover, by Proposition~\ref{prop:right_cover}, $x''_3 = r-\eps \in F(\vec{x}'')$, because for the $(1|2|3,4)$-well-separated instance $\vec{x}$, $F_3(\vec{x}) = x_3$, and $\vec{x}''$ is an $(1|2|3,4)$-well-separated instance with $x''_3 \geq x_3$.
Since both $x''_3, x''_4 \in F(\vec{x}'')$, either the agents $1$ and $2$ are served by the same facility of $F(\vec{x}'')$ or the agent $2$ is served by the facility at $r-\eps$. In both cases, $\cost[F(\vec{x}'')] \geq \lambda$. On the other hand, the optimal cost for $\vec{x}''$ is $\eps \leq 1$, and the approximation ratio of $F$ is at least $\lambda > \rho$. 
\qed\end{proof}

\section{Inexistence of Nice Mechanisms for 2-Facility Location in More General Metrics}
\label{s:tree}

Throughout this section, we consider $3$-location instances of 2-Facility Location with $n \geq 3$ agents in a metric space consisting of $3$ half-lines $[0, \infty)$ with a common origin $O$. This is conceptually equivalent to a continuous metric determined by a star with center $O$ and $3$ long branches starting at $O$. So, we refer to this metric as $S_3$, and to the $3$ half-lines (or branches) of $S_3$ as $b_1$, $b_2$, and $b_3$.
A location $(x, b_\ell)$ in $S_3$ is determined by the distance $x \geq 0$ to the center $O$ and the corresponding branch $b_\ell$, $\ell \in \{1, 2, 3\}$. The distance of two locations $(x, b_\ell)$ and $(x', b_{\ell'})$ in $S_3$ is $|x - x'|$, if $\ell = \ell'$ (i.e., if the locations are on the same branch), and $x+x'$, otherwise. Given two locations $(x, b_\ell)$ and $(x', b_{\ell'})$ in $S_3$, we let $[(x, b_\ell), (x', b_{\ell'})]$ be the interval of all points in the path from $(x, b_\ell)$ to $(x', b_{\ell'})$.

To show that there do not exist any nice mechanisms for 2-Facility Location in $S_3$, we extend Theorem~\ref{thm:3agents} (and Corollary~\ref{cor:3agents}), so that we characterize nice mechanisms for $3$-agent (and $3$-location) instances of 2-Facility Location in $S_3$ when all agents are located on (at most) two fixed branches. As in Section~\ref{s:3locations}, we first extend the characterization to $3$-agent instances, and then use partial group strategyproofness to further extend it to $3$-location instances. We first show that for collinear instances, nice mechanisms do not take any essential advantage of the third branch.

\begin{proposition}\label{prop:tree-collinear}
Let $F$ be a nice mechanism for 2-Facility Location in $S_3$, let $\vec{x}$ be any $3$-agent instance where all agents' locations are on (at most) two branches, and let $f$ be the facility of $F(\vec{x})$ serving (at least) two agents $j$ and $k$ located at $(x_j, b_{\ell_j})$ and $(x_k, b_{\ell_k})$, respectively. Then, $f \in [(x_j, b_{\ell_j}), (x_k, b_{\ell_k})]$.
\end{proposition}

\begin{proof}
Throughout the proof, we let $i$ be the third agent located at $(x_i, b_{\ell_i})$ in $\vec{x}$. For sake of contradiction, we assume that $f$ is not located in $[(x_j, b_{\ell_j}), (x_k, b_{\ell_k})]$.
We first obtain a contradiction in case where $f$, $(x_j, b_{\ell_j})$, and $(x_k, b_{\ell_k})$ are collinear.
If $(x_j, b_{\ell_j})$ is located in $[(x_k, b_{\ell_k}), f]$, then agent $k$ may report $(x_j, b_{\ell_j})$ and decrease her cost, since $(x_j, b_{\ell_j}) \in F(\vec{x}_{-k}, (x_j, b_{\ell_j}))$, due to the bounded approximation ratio of $F$. Similarly, we argue that $(x_k, b_{\ell_k})$ cannot be located in $[(x_j, b_{\ell_j}), f]$. Therefore, if $f$, $(x_j, b_{\ell_j})$, and $(x_k, b_{\ell_k})$ are collinear, then $f \in [(x_j, b_{\ell_j}), (x_k, b_{\ell_k})]$.

We have also to exclude the possibility that the locations $f$, $(x_j, b_{\ell_j})$, and $(x_k, b_{\ell_k})$ are on three different branches. Then, the center $O$ is included in all three paths $[(x_j, b_{\ell_j}), (x_k, b_{\ell_k})]$, $[(x_k, b_{\ell_k}), f]$, and $[f, (x_k, b_{\ell_k})]$. For convenience, we let $(\delta, b_\ell)$ denote the location of the facility $f$ serving $j$ and $k$ in $F(\vec{x})$.

We first observe that $x_j \geq \delta$. Otherwise, the distance $x_k + x_j$ of agent $k$ to agent $j$ would be less than the distance $x_k + \delta$ of agent $k$ to facility $f$, and agent $k$ could  report $(x_j, b_{\ell_j})$ and decrease her cost, since $(x_j, b_{\ell_j}) \in F(\vec{x}_{-k}, (x_j, b_{\ell_j}))$, due to the bounded approximation ratio of $F$. Similarly, we show that $x_k \geq \delta$.

Since the locations of agents $i$, $j$, and $k$ are collinear, we can assume, without loss of generality, that agent $i$ is located on the branch $b_{\ell_j}$ of agent $j$ (the case where $i$ is located on the branch $b_{\ell_k}$ of agent $k$ is symmetric). Moreover, we assume that $x_i > x_j$ (otherwise, $i$ is also served by $f$, and we can switch $i$ and $j$). For some $\eps > 0$ much smaller than $\delta$ and $x_j$, we let $\vec{x}' = (\vec{x}_{-j}, (\eps, b_\ell))$ be the instance obtained from $\vec{x}$ if agent $j$ moves on the branch $b_\ell$ (of $f$) very close to the center $O$. Since $F$ is strategyproof and agent $j$ moves closer to $f$, $F(\vec{x}')$ must have a facility at $f$.
More specifically, since $F$ is strategyproof and $f$ is the facility of $F(\vec{x})$ closest to $(x_j, b_{\ell_j})$, the image set $I_j(\vec{x}_{-j})$ has a hole around $(x_j, b_{\ell_j})$. This hole includes $[(x_j, b_{\ell_j}), (\delta, b_\ell))$ and $[(x_j, b_{\ell_j}), (\delta, b_{\ell_k}))$, since otherwise there must have been a facility closer to $(x_j, b_{\ell_j})$ in $F(\vec{x})$. Therefore, the location in $I_j(\vec{x}_{-j})$ closest to $(\eps, b_\ell)$ is $f$.
Then, if agent $k$ is served by $f$ in $F(\vec{x}')$, she may report $(\eps, b_\ell)$ and decrease her cost from $\delta+x_k$ to $\eps+x_k$, since $(\eps, b_\ell) \in F(\vec{x}'_{-k}, (\eps, b_\ell))$, due to the bounded approximation ratio of $F$. This contradicts the hypothesis that $F$ is strategyproof.

So, let us assume that agent $k$ is served by the second facility of $F(\vec{x}')$. Since the hole around $(x_j, b_{\ell_j})$ in the image set $I_j(\vec{x}_{-j})$ includes $[(x_j, b_{\ell_j}), (\delta, b_{\ell_k}))$, the second facility of $F(\vec{x}')$, that serves agent $k$, must be located in $[(\delta, b_{\ell_k}), (x_k+\eps, b_{\ell_k}))$. Therefore, by the choice of $\eps$, agent $i$ is served by $f$ in $F(\vec{x}')$. Then, agent $i$ may report $(\eps, b_\ell)$ and decrease her cost from $\delta+x_i$ to $\eps+x_i$, since $(\eps, b_\ell) \in F(\vec{x}'_{-i}, (\eps, b_\ell))$, due to the bounded approximation ratio of $F$. This contradicts the strategyproofness of $F$.
\qed\end{proof}

Proposition~\ref{prop:tree-collinear} implies that the characterization of Theorem~\ref{thm:3agents} also applies to nice mechanisms for $3$-agent instances of 2-Facility Location in $S_3$ when all agents are located on (at most) two fixed branches. More specifically, the crux of the proof of Theorem~\ref{thm:3agents} is to argue about all possible strategyproof allocations of the facility serving a pair of agents $j$ and $k$. All these arguments exploit carefully chosen (mostly well-separated) instances of $3$ agents on the line. Proposition~\ref{prop:tree-collinear} shows that for all these instances, a nice mechanism for 2-Facility Location in $S_3$ has to place the facility serving agents $j$ and $k$ in the closed line interval between them, and thus the mechanism cannot take any advantage of the third branch. Therefore, we can restate the whole proof of Theorem~\ref{thm:3agents} with $S_3$ as the underlying metric space, as soon as for all instances considered in the proof, the $3$ agents are located on (at most) two fixed branches.
Moreover, using that any strategyproof mechanism is also partial group strategyproof \cite[Lemma~2.1]{LSWZ10}, we generalize Corollary~\ref{cor:3agents} to $3$-location collinear instances in $S_3$. For convenience, for any pair of branches $b_{\ell_1}$ and $b_{\ell_2}$, we let $\I_{S_3}(b_{\ell_1}, b_{\ell_2})$ be the class of all $3$-location instances in $S_3$ where the agents are located on $b_{\ell_1}$ and $b_{\ell_2}$. We obtain that:

\begin{corollary}\label{cor:3locations-S3}
Let $F$ be a nice mechanism for 2-Facility Location applied to $3$-location instances in $S_3$ with $n \geq 3$ agents, and let $b_{\ell_1}$ and $b_{\ell_2}$ be any two branches of $S_3$. Then, there exist at most two permutations $\pi_1$, $\pi_2$ of the agent coalitions with $\pi_1(2) = \pi_2(2)$, such that for all instances $\vec{x} \in \I_{S_3}(b_{\ell_1}, b_{\ell_2})$ where the agent coalitions are arranged on the line $b_{\ell_1} - b_{\ell_2}$ according to $\pi_1$ or $\pi_2$, $F(\vec{x})$ places a facility at the location of the middle coalition. For any other permutation $\pi$ and instance $\vec{x} \in \I_{S_3}(b_{\ell_1}, b_{\ell_2})$, where the agent coalitions are arranged on the line $b_{\ell_1} - b_{\ell_2}$ according to $\pi$, $F(\vec{x})$ places the facilities at the locations of the leftmost and the rightmost coalitions.
\end{corollary}

As in previous sections, well-separated instances play a crucial role in the proof of the impossibility result for $S_3$. Here we define well-separated instances in a slightly different way. Given a nice mechanism $F$ with approximation ratio $\rho$ for $3$ agents located in $S_3$, we say that a $3$-agent instance $\vec{x}$ is $i$-well-separated, if for some agent $i \in \{1, 2, 3\}$, $2 (\rho+1)$ times the distance of the other two agents $j$ and $k$ is less than the minimum of the distances of $i$ to $j$ and $i$ to $k$. Therefore, due to the approximation ratio of $F$, for any $i$-well-separated instance $\vec{x}$, one facility of $F(\vec{x})$ serves agent $i$ alone and the other facility of $F(\vec{x})$ serves the nearby agents $j$ and $k$.
The following proposition is similar to Proposition~\ref{prop:tree-collinear}, but also applies to non-collinear well-separated instances. It can be regarded as the equivalent of Proposition~\ref{prop:interval} (and of \cite[Lemma~2]{SV02}) for $3$-agent instances in $S_3$.

\begin{proposition}\label{prop:tree-well-sep}
Let $F$ be a nice mechanism for 2-Facility Location with $3$-agents in $S_3$. For any $i$-well-separated instance $\vec{x}$, the facility of $F(\vec{x})$ serving the two nearby agents $j$ and $k$ is in $[(x_j, b_{\ell_j}), (x_k, b_{\ell_k})]$.
\end{proposition}

\begin{proof}
Let $f$ be the facility of $F(\vec{x})$ that serves the two nearby agents $j$ and $k$, and let $\rho$ be the approximation ratio of $F$.
For sake of contradiction, we assume that $f$ is not located in $[(x_j, b_{\ell_j}), (x_k, b_{\ell_k})]$. In case where $f$, $(x_j, b_{\ell_j})$, and $(x_k, b_{\ell_k})$ are collinear, we obtain a contradiction as in the corresponding case in the proof of Proposition~\ref{prop:tree-collinear}.
%
%We first obtain a contradiction in case where $f$, $(x_j, b_{\ell_j})$, and $(x_k, b_{\ell_k})$ are collinear.
%
%If $(x_j, b_{\ell_j})$ is located in $[(x_k, b_{\ell_k}), f]$, then agent $k$ may report $(x_j, b_{\ell_j})$ and decrease her cost, since $(x_j, b_{\ell_j}) \in F(\vec{x}_{-k}, (x_j, b_{\ell_j}))$, due to the bounded approximation ratio of $F$. Similarly, we argue that $(x_k, b_{\ell_k})$ cannot be located in $[(x_j, b_{\ell_j}), f]$. Therefore, if $f$, $(x_j, b_{\ell_j})$, and $(x_k, b_{\ell_k})$ are collinear, $f \in [(x_j, b_{\ell_j}), (x_k, b_{\ell_k})]$.
%
We have also to exclude the possibility that the locations $f$, $(x_j, b_{\ell_j})$, and $(x_k, b_{\ell_k})$ are on three different branches. Then, the center $O$ is included in all three paths $[(x_j, b_{\ell_j}), (x_k, b_{\ell_k})]$, $[(x_k, b_{\ell_k}), f]$, and $[f, (x_k, b_{\ell_k})]$. For convenience, we let $(\delta, b_\ell)$ be the location of the facility $f$ serving $j$ and $k$ in $F(\vec{x})$.

For some $\eps > 0$ much smaller than $\delta$ and $x_j$, we let $\vec{x}' = (\vec{x}_{-j}, (\eps, b_\ell))$ be the instance obtained from $\vec{x}$ if agent $j$ moves on the branch $b_\ell$ and very close to $O$. As in the proof of Proposition~\ref{prop:tree-collinear}, we show that $F(\vec{x}')$ has a facility at $f$. Moreover, due to the hypothesis that $F$ has an approximation ratio of at most $\rho$, the facility at $f$ serves both agents $j$ and $k$.
In particular, if agent $i$ is located either on $b_{\ell_j}$ or on $b_{\ell_k}$, this holds because the instance $\vec{x}'$ is $i$-well-separated because the distance of $i$ to the nearest of $j$ and $k$ does not decrease when $j$ moves to $(\eps, b_\ell)$, while the distance of $j$ to $k$ decreases.
On the other hand, if $i$ is located on $b_\ell$, we have that $2(\rho+1)(x_j+x_j) < x_i + x_j + x_k$, because the instance $\vec{x}$ is $i$-well-separated, and the distance of $i$ to the nearest of $j$ and $k$ in $\vec{x}'$ is $x_i-\eps$ and the distance of $j$ to $k$ is $x_k + \eps$. Using $2(\rho+1)(x_j+x_j) \leq x_i + x_j + x_k$ and $\eps < x_j$, we obtain that $2\rho(x_k+\eps) < x_i -\eps$. Therefore, by the hypothesis that $F$ has an approximation ratio of at most $\rho$, the facility at $f$ serves agents $j$ and $k$ and the other facility of $F(\vec{x}')$ serves $i$ alone.
Then, agent $k$ may report $(\eps, b_\ell)$ and decrease her cost from $\delta+x_k$ to $\eps+x_k$, since $(\eps, b_\ell) \in F(\vec{x}'_{-k}, (\eps, b_\ell))$, due to the bounded approximation ratio of $F$. This contradicts the strategyproofness of $F$.
 \qed\end{proof}

We also need the following proposition, which can be regarded as the equivalent of Proposition~\ref{prop:b_pushes_c} and Proposition~\ref{prop:c_pushes_b} for well-separated instances in $S_3$.

\begin{proposition}\label{prop:tree-move}
Let $F$ be any nice mechanism for 2-Facility Location with $3$-agents in $S_3$, let $\vec{x}$ be any $i$-well-separated instance where the nearby agents $j$ and $k$ are located on different branches $b_{\ell_j}$ and $b_{\ell_k}$, and $2(\rho+1)(x_j+x_k) < x_i$, with $\rho$ denoting the approximation ratio of $F$. If $(x_j, b_{\ell_j}) \in F(\vec{x})$, then for every location $(x'_j, b_{\ell_j})$, with $x'_j \in [0, x_j]$, it holds that $(x'_j, b_{\ell_j}) \in F(\vec{x}_{-j}, (x'_j, b_{\ell_j}))$.
\end{proposition}

\begin{proof}
For sake of contradiction, we assume that there is a $y \in [0, x_j)$ such that $(y, b_{\ell_j}) \not\in F(\vec{x}_{-j}, (y, b_{\ell_j}))$. Hence $(y, b_{\ell_j}) \not\in I_j(\vec{x}_{-j})$, and there is a hole in the image set $I_j(\vec{x}_{-j})$. Let $(z, b_{\ell_j})$, $z > y$, be the location in $I_j(\vec{x}_{-j})$ on the branch $b_{\ell_j}$ closest to $(y, b_{\ell_j})$.
Such a location exists because $(x_j, b_{\ell_j}) \in F(\vec{x})$. For some very small $\eps \in (0, z)$, we let $\vec{x}' = (\vec{x}_{-j}, (z-\eps, b_{\ell_j}))$ be the instance obtained from $\vec{x}$ if agent $j$ moves on the branch $b_{\ell_j}$ just before $z$. Since $F$ is strategyproof, $(z, b_{\ell_j}) \in F(\vec{x}')$. Moreover, the hypothesis that $2(\rho+1)(x_j+x_k) < x_i$ implies that the instance $\vec{x}'$ is $i$-well separated. Therefore, due to the bounded approximation ratio of $F$, the facility at $(z, b_{\ell_j})$ serves both agents $j$ and $k$ in $F(\vec{x}')$. Thus, we obtain a contradiction, since by Proposition~\ref{prop:tree-well-sep}, the facility at $(z, b_{\ell_j})$ must be located in $[(x_k, b_{\ell_k}), (z-\eps, b_{\ell_j})]$.
\qed\end{proof}

Using Corollary~\ref{cor:3locations-S3}, Proposition~\ref{prop:tree-well-sep}, and Proposition~\ref{prop:tree-move}, we next show that there are no nice mechanisms for 2-Facility Location with $3$ agents in $S_3$.

\begin{theorem}\label{thm:tree-unbounded}
Any deterministic strategyproof mechanism for 2-Facility Location with $3$ agents in $S_3$ has an unbounded approximation ratio.
\end{theorem}

\begin{proof}
For sake of contradiction, we let $F$ be a nice mechanism for 2-Facility Location with $3$ agents in $S_3$. Applying Corollary~\ref{cor:3locations-S3} to $F$ and all instances with $3$ agents located on the branches $b_1$ and $b_2$, we obtain that there exists an agent $i$, such that for all instances $\vec{x} \in \I_{S_3}(b_1, b_2)$, if $i$ is located at the one extreme, $F(\vec{x})$ places the facilities at the two extreme locations of $\vec{x}$ (if $F$ does not admit a partial dictator on the line $b_1 - b_2$, $i$ can be any agent). Similarly, there exists an agent $k$, which may be the same as or different from $i$, such that for all $3$-agent instances $\vec{x} \in \I_{S_3}(b_1, b_3)$, if $k$ is located at the one extreme, $F(\vec{x})$ places the facilities at the two extreme locations of $\vec{x}$.

In the following, we let $i$ be the partial dictator of $F$ on the line $b_1 - b_2$ (or any agent, if $F$ does not admit a partial dictator on $b_1 - b_2$), let $k$ be the partial dictator of $F$ on the line $b_1 - b_3$, if this agent is different from $i$, or any agent different from $i$ otherwise, and let $j$ be the third agent. For some $\delta > 0$, we consider an instance $\vec{x}$ where the locations of agents $j$ and $k$ are $(\delta, b_2)$ and $(\delta, b_3)$, respectively, and the location of agent $i$ is $(a, b_1)$, where $a > 4(\rho+1) \delta$ is chosen so that $\vec{x}$ is $i$-well-separated and we can apply Proposition~\ref{prop:tree-move}. Therefore, there is a facility $f \in F(\vec{x})$ that serves both $j$ and $k$, and by Proposition~\ref{prop:tree-well-sep}, $f \in [(\delta, b_2), (\delta, b_3)]$.

Let us first assume that $f \in [(\delta, b_2), (0, b_2)]$ (note that the location $(0, b_2)$ coincides with the center $O$ of $S_3$). Then, we use Proposition~\ref{prop:tree-move} and show that for the instance $\vec{x}' = (\vec{x}_{-j}, (0, b_2))$, where all agents are located on the line $b_1 - b_3$, $(0, b_2) \in F(\vec{x}')$, i.e. $F(\vec{x}')$ does not place the facilities at the two extreme locations of $\vec{x}'$. This is a contradiction, since by the choice of $k$ as the (possible) partial dictator of $F$ on the line $b_1 - b_3$, for all instances $\vec{y} \in \I_{S_3}(b_1, b_3)$, $F(\vec{y})$ must place the facilities at the two extreme locations of $\vec{y}$.
More specifically, if $f$ is located at $(\delta, b_2)$, i.e., at the location of agent $j$ in $\vec{x}$, Proposition~\ref{prop:tree-move} immediately implies that $(0, b_2) \in F(\vec{x}_{-j}, (0, b_2))$. Otherwise, we let $(\delta', b_2)$, with $\delta' \in (0, \delta)$, be the location of $f$, and consider the instance $\vec{x}'' = (\vec{x}_{-j}, (\delta', b_2))$, which is also $i$-well-separated and satisfies the hypothesis of Proposition~\ref{prop:tree-move}, due to the choice of $a$. Since $F$ is strategyproof, $F(\vec{x}'')$ has a facility at $(\delta', b_2)$. Then,  Proposition~\ref{prop:tree-move} implies that $(0, b_2) \in F(\vec{x}''_{-j}, (0, b_2))$.

Therefore, the facility of $F(\vec{x})$ serving $j$ and $k$ cannot be located in $[(\delta, b_2), (0, b_2)]$. By the same argument, but with agent $k$ in place of agent $j$, and agent $i$ as the (possible) partial dictator of $F$ on the line $b_1 - b_2$, we show that the facility of $F(\vec{x})$ serving $j$ and $k$ cannot be located in $[(\delta, b_3), (0, b_3)]$. Thus, we obtain a contradiction, and conclude the proof of the theorem.
\qed\end{proof}

We can generalize the proof of Theorem~\ref{thm:tree-unbounded} to instances with $n \geq 3$ agents. To this end, we start with a $3$-location instance $\vec{x}$ with $n \geq 3$ agents, where a coalition of $n - 2$ agents, including the partial dictator of $F$ on the line $b_1 - b_2$, plays the role of agent $i$, and the remaining $2$ agents play the role of agents $j$ and $k$ in the proof of Theorem~\ref{thm:tree-unbounded}. Then, using Corollary~\ref{cor:3locations-S3} for $3$-location instances in $S_3$ and the fact that any strategyproof mechanism is also partial group strategyproof, we can restate the proofs of Proposition~\ref{prop:tree-well-sep}, Proposition~\ref{prop:tree-move}, and Theorem~\ref{thm:tree-unbounded} for such $3$-location instances in $S_3$, and obtain that:

\begin{corollary}\label{cor:tree-unbounded}
Any deterministic strategyproof mechanism for 2-Facility Location with $n \geq 3$ agents in $S_3$ has an unbounded approximation ratio.
\end{corollary}

\section{Discussion and Open Problems}

An open problem is whether one can use the techniques in the proof of Theorem~\ref{thm:3agents}, and extend the impossibility result of Theorem~\ref{thm:3facilities} to non-anonymous mechanisms.
Two other intriguing directions for research have to do with the approximability of $K$-Facility Location, for $K \geq 4$, by randomized and deterministic imposing mechanisms. For $3$-Facility Location on the line, there are both a randomized and a deterministic imposing mechanism with approximation ratio $n-1$.
%
%Furthermore, for all $K \geq 3$, {\sc Inversely Proportional} is strategyproof and achieves an approximation ratio of $(K+1)/2$ for instances with $K+1$ agents.
%
Therefore, it is very interesting to obtain a deterministic imposing (resp. randomized) mechanism with a bounded approximation ratio for $K$-Facility Location, for all $K \geq 4$ (resp. and all $n \geq K+2)$, or to show that no such mechanism exists.
Moreover, we are aware of a deterministic imposing mechanism for $2$-Facility Location on the line that escapes the characterization of Theorem~\ref{thm:2fac-gen}, albeit with an approximation ratio of $n-1$. However, this raises the question about the approximability of $2$-Facility and $3$-Facility Location on the line by deterministic imposing mechanisms.

\bibliographystyle{plain}
\bibliography{mechdesign}

\newpage\appendix
\section{Appendix~\ref{app:s:moves}: Dealing with Well-Separated Instances}
\label{app:s:moves}

\subsection{Pushing the Pair of the Rightmost Agents to the Right: The Proof of Proposition~\ref{prop:right_cover}}
\label{app:s:push_right}

For simplicity and clarity, we explicitly prove Proposition~\ref{prop:right_cover} only for $2$-Facility Location and well-separated instances with $3$ agents (see Proposition~\ref{prop:right_cover_3agents}). It is not difficult to verify that all the technical arguments only depend on the fact that the two nearby agents stay well-separated from and on the right of the third agent. Therefore, the following propositions generalize to the $K$-Facility Location game and well-separated instances with $K+1$ agents.

In the following, we let $F$ be a nice mechanism for the $2$-Facility Location game with an approximation ratio of at most $\rho$ for instances with $3$ agents.
As in Section~\ref{s:3agents}, we use the indices $i, j, k$ to implicitly define a permutation of the agents. We use the convention that $i$ denotes the leftmost agent, $j$ denotes the middle agent, and $k$ denotes the rightmost agent.
We recall that given a nice mechanism $F$ with an approximation ratio of at most $\rho$ for $3$-agent instances, a $3$-agent instance $\vec{x}$ is \emph{$(i | j, k)$-well-separated} if $x_i < x_j < x_k$ and $\rho(x_k - x_j) < x_j - x_i$.

The following propositions show that if for some nice mechanism $F$, there is an $(i | j, k)$-well-separated instance $\vec{x}$ such that $F_2(\vec{x}) = x_j$, then as long as we ``push'' the locations of agents $j$ and $k$ to the right, while keeping the instance $(i | j, k)$-well-separated, the rightmost facility of $F$ stays with the location of agent $j$.
Intuitively, if for some $(i | j, k)$-well-separated instance $\vec{x}$, $F_2(\vec{x}) = x_j$, then agent $j$ serves as a dictator for all $(i | j, k)$-well-separated instances $\vec{x}' = (\vec{x}_{-\{j,k\}}, x'_j, x'_k)$ with $x_j \leq x'_j$.

\begin{proposition}\label{prop:b_pushes_c}
Let $\vec{x}$ be any $(i | j, k)$-well-separated instance with $F_2(\vec{x}) = x_j$. Then for every instance $\vec{x}' = (\vec{x}_{-j}, x'_j)$ with $x_j < x'_j < x_k$, it holds that $F_2(\vec{x}') = x'_j$.
\end{proposition}

\begin{proof}
We first observe that any instance $\vec{x}' = (\vec{x}_{-j}, x'_j)$, with $x_j < x'_j < x_k$, is also $(i | j, k)$-well-separated.
To reach a contradiction, we assume that there is a point $y \in (x_j, x_k)$ such that $y \neq F_2(\vec{x}_{-j}, y)$. Hence $y \not\in I_j(x_i, x_k)$, and there is a $y$-hole $(l, r)$ in the image set $I_j(x_i, x_k)$. Let $y'$ be any point in the left half of $(l, r)$, e.g. let $y' = (2l+r)/3$. Then $l \in F(\vec{x}_{-j}, y')$. By Proposition~\ref{prop:middle}, $F_2(\vec{x}_{-j}, y') > y'$. This contradicts $F$'s bounded approximation ratio, since both $x_i$ and $y'$ are served by the facility at $l$, and $\cost[F(\vec{x}_{-j}, y')] \geq y' - x_i$, while the optimal cost is $x_k - y' > \rho(y' - x_i)$, because the instance $(\vec{x}_{-j}, y')$ is $(i | j, k)$-well-separated.
\qed \end{proof}

\begin{proposition}\label{prop:c_pulls_b}
Let $\vec{x}$ be any $(i | j, k)$-well-separated instance with $F_2(\vec{x}) = x_j$. Then for every $(i | j, k)$-well-separated instance $\vec{x}' = (\vec{x}_{-k}, x'_k)$, $F_2(\vec{x}') = x_j$.
\end{proposition}

\begin{proof}
Since $F_2(\vec{x}) < x_k$, $x_k$ does not belong to the image set $I_k(x_i, x_j)$, and there is a $x_k$-hole $(l, r)$ in $I_k(x_i, x_j)$. Since $F_2(\vec{x}) = x_j$, the left endpoint of the $x_k$-hole is $l = x_j$ and the right endpoint is $r \geq 2x_k - x_j$.
Therefore, for all $(i | j, k)$-well-separated instances $\vec{x}' = (\vec{x}_{-k}, x'_k)$ with $x'_k < (r + l)/2$, $F_2(\vec{x}') = x_j$.

To conclude the proof, we show that there are no $(i | j, k)$-well-separated instances $\vec{x}' = (\vec{x}_{-k}, x'_k)$ with $x'_k \geq (r + l)/2$ and $F_2(\vec{x}') \neq x_j$.
To reach a contradiction, we assume that there exists a point $y \geq (r + l)/2$ such that $(\vec{x}_{-k}, y)$ is an $(i | j, k)$-well-separated instance and $F_2(\vec{x}_{-k}, y) \neq x_j$.
The existence of such a point $y$ implies the existence of a point $x'_k \in [(r + l)/2, r)$ ($x'_k$ may coincide with $y$) for which $\vec{x}' = (\vec{x}_{-k}, x'_k)$ is an $(i | j, k)$-well-separated instance. Then, $F_2(\vec{x}') = r > x'_k$, because the distance of $x'_k$ to $r \in I_k(x_i, x_j)$ is no greater than the distance of $x'_k$ to $l$.
Since $\vec{x}'$ is an $(i | j, k)$-well-separated instance, this contradicts Proposition~\ref{prop:interval}, according to which $F_2(\vec{x}) \in [x'_j, x'_k]$.
%
%Therefore, by Proposition~\ref{prop:middle}, $F_1(\vec{x}') \leq x_j$. Moreover, $x_j - F_1(\vec{x}') \leq x'_k - x_j$. Otherwise, agent $j$ has an incentive to report $x'_k$ and decrease his cost, since $x'_k \in F(x_i, x'_k, x'_k)$ due to the bounded approximation ratio of $F$.
%
%However, the fact that both $x_i$ and $x_j$ are served by $F_1(\vec{x}')$ contradicts $F$'s bounded approximation ratio, since $\cost[F(\vec{x}')] \geq x_j - x_i$, while the optimal cost is $x'_k - x_j < (x_j - x_i)/\rho$, because the instance $(\vec{x}_{-k}, x'_k)$ is $(i | j, k)$-well-separated.
\qed \end{proof}

\begin{proposition}\label{prop:right_cover_step}
Let $\vec{x}$ be any $(i | j, k)$-well-separated instance with $F_2(\vec{x}) = x_j$. For every $(i | j, k)$-well-separated instance $\vec{x}' = (\vec{x}_{-\{j, k\}}, x'_j, x'_k)$ with $x_j < x'_j \leq (x_j + x_k)/2$, it holds that $F_2(\vec{x}') = x'_j$.
\end{proposition}

\begin{proof}
Since $x'_j \leq (x_j + x_k) / 2 < x_k$, by Proposition~\ref{prop:b_pushes_c}, $F_2(\vec{x}_{-j}, x'_j) = x'_j$. Since the distance of $x'_j$ to $x_k$ is smaller than the distance of $x_j$ to $x_k$, the new instance $(\vec{x}_{-j}, x'_j)$ is $(i | j, k)$-well-separated. Therefore, by Proposition~\ref{prop:c_pulls_b}, for any $(i | j, k)$-well-separated instance $\vec{x}' = (\vec{x}_{-\{j, k\}}, x'_j, x'_k)$, $F_2(\vec{x}') = x'_j$.
\qed \end{proof}

\begin{proposition}\label{prop:right_cover_3agents}
Let $\vec{x}$ be any $(i | j, k)$-well-separated instance with $F_2(\vec{x}) = x_j$. Then for every $(i | j, k)$-well-separated instance $\vec{x}' = (\vec{x}_{-\{j, k\}}, x'_j, x'_k)$ with $x_j \leq x'_j$, it holds that $F_2(\vec{x}') = x'_j$.
\end{proposition}

\begin{proof}
The proof follows by an inductive application of Proposition~\ref{prop:right_cover_step}. More specifically, by repeated applications of Proposition~\ref{prop:right_cover_step}, we keep moving the locations of agents $j$ and $k$ to the right, while keeping the resulting instance $(i | j, k)$-well-separated, and thus maintaining the location of the rightmost facility at the location of agent $j$.

Formally, let $d = x'_j - x_j$, let $\delta = (x_k - x_j)/2$, and let $\kappa = \ceil{d/\delta}$.
For every $\lambda = 1, \ldots, \kappa$, we inductively consider the instance $\vec{x}_\lambda = (\vec{x}_{-\{j, k\}}, x_j + (\lambda-1) \delta, x_k + (\lambda-1)\delta)$.
We observe that the instance $\vec{x}_\lambda$ is $(i | j, k)$-well-separated, because the distance of the locations of agents $j$ and $k$ is $2\delta$, while the distance of the locations of agents $i$ and $j$ is at least their distance in $\vec{x}$.
By inductively applying Proposition~\ref{prop:right_cover_step} to $\vec{x}_\lambda$, we obtain that for every $(i | j, k)$-well-separated instance $(\vec{x}_{-\{j,k\}}, y_j, y_k)$ with $x_j + (\lambda-1) \delta \leq y_j \leq x_j + \lambda \delta$, $F_2(\vec{x}_{-\{j,k\}}, y_j, y_k) = y_j$. For $\lambda = \kappa$, we conclude that $F_2(\vec{x}_{-\{j, k\}}, x'_j, x'_k) = x'_j$.
\qed\end{proof}

\subsection{Pushing the Pair of the Rightmost Agents to the Left: The Proof of Proposition~\ref{prop:left_cover}}
\label{app:s:push_left}

As in Section~\ref{app:s:push_right}, we explicitly prove Proposition~\ref{prop:left_cover} only for $2$-Facility Location and well-separated instances with $3$ agents (see also Proposition~\ref{prop:left_cover_3agents}). As before, it is not difficult to verify that the following propositions generalize to the $K$-Facility Location game and well-separated instances with $K+1$ agents.
%
%Moreover, one can establish the equivalent of the following propositions for well-separated instances where the two nearby agents are located elsewhere in the instance.

Next, we use the same notation as in Section~\ref{app:s:push_right}.
The following propositions show that if for some nice mechanism $F$, there is an $(i | j, k)$-well-separated instance $\vec{x}$ such that $F_2(\vec{x}) = x_k$, then as long as we ``push'' the locations of agents $j$ and $k$ to the left, while keeping the instance $(i | j, k)$-well-separated, the rightmost facility of $F$ stays with the location of agent $k$.
Intuitively, if for some $(i | j, k)$-well-separated instance $\vec{x}$, $F_2(\vec{x}) = x_k$, then agent $k$ serves as a dictator for all $(i | j, k)$-well-separated instances $\vec{x}' = (\vec{x}_{-\{j,k\}}, x'_j, x'_k)$ with $x'_k \leq x_k$.

\begin{proposition}\label{prop:c_pushes_b}
Let $\vec{x}$ be any $(i | j, k)$-well-separated instance with $F_2(\vec{x}) = x_k$. Then for every instance $\vec{x}' = (\vec{x}_{-k},x'_k)$ with $x_j < x'_k < x_k$, it holds that $F_2(\vec{x}') = x'_k$.
\end{proposition}

\begin{proof}
To reach a contradiction, we assume that there is a point $x'_k \in (x_j, x_k)$ such that $x'_k \neq F_2(\vec{x}_{-k}, x'_k)$. Therefore, $x'_k \not\in I_k(x_i, x_j)$, and there is a $x'_k$-hole $(l, r)$ in the image set $I_k(x_i, x_j)$. We observe that $x_j \leq l$, since $x_j \in F(\vec{x}_{-k}, x_j)$, due to the bounded approximation ratio of $F$, and that $r \leq x_k$, since $F_2(\vec{x}) = x_k$.
Let $y_k$ be any point in the right half of $(l, r)$ different from $r$, e.g. let $y_k = (l+2r)/3$. We consider the instance $\vec{y} = (\vec{x}_{-k}, y_k)$, for which $F_2(\vec{y}) = r$, due to the strategyproofness of $F$.
Since $x_j < y_k < x_k$, $\vec{y}$ is an $(i | j, k)$-well-separated instance. Therefore, $F_2(\vec{y}) = r > y_k$ contradicts Proposition~\ref{prop:interval}, according to which $F_2(\vec{y}) \in [y_j, y_k]$.
%
%By Proposition~\ref{prop:middle}, $F_1(\vec{x}_{-k}, y') \leq x_j$. Moreover, $x_j - F_1(\vec{x}_{-k}, y') \leq y' - x_j$. Otherwise, agent $j$ has an incentive to report $y'$ and decrease his cost, since $y' \in F(\vec{x}_{-\{j,k\}}, y', y')$ due to the bounded approximation ratio of $F$.
%
%However, the fact that both $x_i$ and $x_j$ are served by $F_1(\vec{x}_{-k}, y')$ contradicts $F$'s bounded approximation ratio, since $\cost[F(\vec{x}_{-k}, y')] \geq x_j - x_i$, while the optimal cost is $y' - x_j > \rho(x_j - x_i)$, because $x_j < y' < x_k$, and thus the instance $(\vec{x}_{-k}, y')$ is $(i | j, k)$-well-separated.
\qed \end{proof}

\begin{proposition}\label{prop:b_pulls_c}
Let $\vec{x}$ be any $(i | j, k)$-well-separated instance with $F_2(\vec{x}) = x_k$. Then for every $(i | j, k)$-well-separated instance $\vec{x}' = (\vec{x}_{-j}, x'_j)$, $F_2(\vec{x}') = x_k$.
\end{proposition}

\begin{proof}
Since $\vec{x}$ is $(i | j, k)$-well-separated, $F_1(\vec{x}) < x_j$ due to $F$'s bounded approximation ratio. Therefore, $x_j$ does not belong to the image set $I_j(x_i, x_k)$, and there is a $x_j$-hole $(l, r)$ in $I_j(x_i, x_k)$. Since $F_2(\vec{x}) = x_k$, the right endpoint of the $x_k$-hole is $r = x_k$ and the left endpoint is $l \leq 2x_j - x_k$.
Therefore, for all $(i | j, k)$-well-separated instances $\vec{x}' = (\vec{x}_{-j}, x'_j)$ with $x'_j > (r + l)/2$, $F_2(\vec{x}') = x_k$.

Next, we show that there are no $(i | j, k)$-well-separated instances $\vec{x}' = (\vec{x}_{-j}, x'_j)$ with $x'_j \leq (r + l)/2$ and $F_2(\vec{x}') \neq x_k$. To reach a contradiction, we assume that there exists a point $y \leq (r + l)/2$ such that $(\vec{x}_{-j}, y)$ is $(i | j, k)$-well-separated and $F_2(\vec{x}_{-j}, y) \neq x_k$.
The existence of such a point $y$ implies the existence of a point $x'_j \in (l, (r + l)/2]$ ($x'_j$ may coincide with $y$) for which $\vec{x}' = (\vec{x}_{-j}, x'_j)$ is an $(i | j, k)$-well-separated instance. Therefore, $l \in F(\vec{x}')$, because the distance of $x'_j$ to $l \in I_j(x_i, x_k)$ is no greater than the distance of $x'_j$ to $r$.
By Proposition~\ref{prop:middle}, $F_2(\vec{x}') \geq x'_j$, which implies that $F_2(\vec{x}') \geq x_k$. Since $x'_j$ lies in the left half of $(l, r)$, both $x_i$ and $x'_j$ are served by the facility at $l$, which contradicts $F$'s bounded approximation ratio, because $\cost[F(\vec{x}')] \geq x'_j - x_i$, while the optimal cost is $x_k - x'_j < (x'_j - x_i)/\rho$, because the instance $\vec{x}'$ is $(i | j, k)$-well-separated.
\qed \end{proof}

\begin{proposition}\label{prop:left_cover_step}
Let $\vec{x}$ be any $(i | j, k)$-well-separated instance with $F_2(\vec{x}) = x_k$. Then for every $(i | j, k)$-well-separated instance $\vec{x}' = (\vec{x}_{-\{j,k\}}, x'_j, x'_k)$ with $(x_k + x_j) / 2 \leq x'_k < x_k$, it holds that $F_2(\vec{x}') = x'_k$.
\end{proposition}

\begin{proof}
Since $x_j < x'_k < x_k$, by Proposition~\ref{prop:c_pushes_b}, $F_2(\vec{x}_{-k}, x'_k) = x'_k$. Since the distance of $x'_k$ to $x_j$ is smaller than the distance of $x_k$ to $x_j$, the new instance $(\vec{x}_{-k}, x'_k)$ is $(i | j, k)$-well-separated. Therefore, by Proposition~\ref{prop:b_pulls_c}, for any $(i | j, k)$-well-separated instance $\vec{x}' = (\vec{x}_{-\{j,k\}}, x'_j, x'_k)$, $F_2(\vec{x}') = x'_k$.
\qed \end{proof}

\begin{proposition}\label{prop:left_cover_3agents}
Let $\vec{x}$ be any $(i | j, k)$-well-separated instance with $F_2(\vec{x}) = x_k$. Then for every $(i | j, k)$-well-separated instance $\vec{x}' = (\vec{x}_{-\{j,k\}}, x'_j, x'_k)$ with $x'_k \leq x_k$, it holds that $F_2(\vec{x}') = x'_k$.
\end{proposition}

\begin{proof}
The proof follows by an inductive application of Proposition~\ref{prop:left_cover_step}.
Let $d = x'_k - x_k$, let $\delta = (x'_k - x'_j)/2$, and let $\kappa = \ceil{d/\delta}$. We first observe that $F_2(\vec{x}_{-j}, x_k - 2\delta) = x_k$, by Proposition~\ref{prop:b_pulls_c}, since the instance $(\vec{x}_{-j}, x_k - 2\delta)$ is $(i | j, k)$-well-separated.
Next, for every $\lambda = 1, \ldots, \kappa$, we inductively consider the instance $\vec{x}_\lambda = (\vec{x}_{-\{j, k\}}, x_k - (\lambda+1) \delta, x_k - (\lambda-1)\delta)$.
We observe that the instance $\vec{x}_\lambda$ is $(i | j, k)$-well-separated, because the distance of the locations of agents $j$ and $k$ is $2\delta$, while the distance of the locations of agents $i$ and $j$ is at least their distance in $\vec{x}'$.
By inductively applying Proposition~\ref{prop:left_cover_step} to $\vec{x}_\lambda$, we obtain that for every $(i | j, k)$-well-separated instance $(\vec{x}_{-\{j,k\}}, y_j, y_k)$ with $x_k - \lambda \delta \leq y_k \leq x_k - (\lambda-1) \delta$, $F_2(\vec{x}_{-\{j,k\}}, y_j, y_k) = y_k$.
For $\lambda = \kappa$, we obtain that $F_2(\vec{x}_{-\{j, k\}}, x'_j, x'_k) = x'_k$.
\qed\end{proof}

\end{document}